\documentclass[a4paper, UKenglish, cleveref, autoref, thm-restate]{lipics-v2021}
\usepackage{csquotes}
\usepackage{amsmath}
\usepackage{amsthm}
\usepackage{algorithm}
\usepackage{algpseudocode}
\usepackage{graphicx}
\usepackage{float}
\usepackage{caption}
\usepackage{subcaption}
\usepackage{todonotes}
\usepackage{tikz}
\usetikzlibrary{calc}
\usepackage{verbatim}
\usepackage{mathtools}
\usepackage{standalone}
\usepackage{pgfplots}
\usepackage{hyperref}
\usepackage[capitalise]{cleveref}
\usepackage{physics}
\usepackage{thmtools, thm-restate}
\usepackage{xparse}
\usepackage{acronym}

\DeclarePairedDelimiterX\setImpl[1]{\{}{\}}{#1}
\NewDocumentCommand\set{sO{}m>{\TrimSpaces}o}{
  \IfValueTF{#4}{
    \setImpl[#2]{#3 : \IfBooleanTF{#1}{\text{#4}}{#4}}
  }{
    \setImpl[#2]{#3}
  }
}

\definecolor{blue}{HTML}{648FFF}
\definecolor{red}{HTML}{DC267F}
\definecolor{yellow}{HTML}{FFB000}
\definecolor{orange}{HTML}{FE6100}

\definecolor{vi}{HTML}{39842E}
\definecolor{li}{HTML}{00677C}
\definecolor{si}{HTML}{E43117}
\definecolor{hi}{HTML}{9B0a7d}
\tikzset{font=\LARGE}

\newtheorem{assumption}{Assumption}[section]

\theoremstyle{remark}

\newtheoremstyle{case}{}{}{}{}{}{:}{ }{}
\theoremstyle{case}
\newtheorem{case}{Case}

\newtheoremstyle{subcase}{}{}{}{}{}{:}{ }{}
\theoremstyle{subcase}
\newtheorem{subcase}{Case}[case]

\newtheoremstyle{subsubcase}{}{}{}{}{}{:}{ }{}
\theoremstyle{subsubcase}
\newtheorem{subsubcase}{Subcase}[subcase]

\acrodef{PTAS}[\textsc{PTAS}]{Polynomial Time Approximation Scheme}
\acrodef{FPTAS}[\textsc{FPTAS}]{Fully Polynomial Time Approximation Scheme}
\acrodef{MCKP}[\textsc{MCKP}]{Multiple-Choice Knapsack Problem}

\title{A Practical 73/50 Approximation for Contiguous Monotone Moldable Job Scheduling}

\titlerunning{Moldable Job Scheduling}

\author{Klaus Jansen}{Kiel University, Germany}{kj@informatik.uni-kiel.de}{https://orcid.org/0000-0001-8358-6796}{}
\author{Felix Ohnesorge}{Kiel University, Germany}{foh@informatik.uni-kiel.de}{https://orcid.org/0009-0003-8023-3380}{}

\authorrunning{K. Jansen, F. Ohnesorge}

\Copyright{Klaus Jansen, Felix Ohnesorge}

\ccsdesc[500]{Theory of computation~Scheduling algorithms}
\ccsdesc[500]{Theory of computation~Discrete optimization}
\ccsdesc[300]{Theory of computation~Parallel computing models}

\keywords{computing, machine scheduling, moldable, polynomial approximation} 

\category{} 

\supplementdetails{Software}{https://github.com/Felioh/MoldableJobScheduling}

\relatedversion{} 

\funding{Funded by the Deutsche Forschungsgemeinschaft (DFG, German Research Foundation) - Project number 453769249}

\hideLIPIcs
\nolinenumbers

\EventEditors{}
\EventNoEds{0}
\EventLongTitle{}
\EventShortTitle{}
\EventAcronym{}
\EventYear{}
\EventDate{}
\EventLocation{}
\EventLogo{}
\SeriesVolume{}
\ArticleNo{}

\begin{document}
\maketitle

\begin{abstract}
    In moldable job scheduling, we are provided $m$ identical machines and $n$ jobs that can be executed on a variable number of machines. The execution time of each job depends on the number of machines assigned to execute that job.
    For the specific problem of \emph{monotone} moldable job scheduling, jobs are assumed to have a processing time that is non-increasing in the number of machines.

    The previous best-known algorithms are: (1) a polynomial-time approximation scheme with time complexity $\Omega(n^{g(1/\varepsilon)})$, where $g(\cdot)$ is a super-exponential function [Jansen and Thöle '08; Jansen and Land '18], (2) a fully polynomial approximation scheme for the case of $m \geq 8\frac{n}{\varepsilon}$ [Jansen and Land '18], and (3) a $\frac{3}{2}$ approximation with time complexity $O(nm\log(mn))$ [Wu, Zhang, and Chen '23].

    We present a new practically efficient algorithm with an approximation ratio of $\approx (1.4593 + \varepsilon)$ and a time complexity of $O(nm \log \frac{1}{\varepsilon})$.
    Our result also applies to the \emph{contiguous} variant of the problem.
    In addition to our theoretical results, we implement the presented algorithm and show that the practical performance is significantly better than the theoretical worst-case approximation ratio.
\end{abstract}

\newpage
\section{Introduction}

In the past many variants of scheduling problems have been studied.
The classical task scheduling problem consists of scheduling \(n\) jobs to \(m\) identical machines, where each job \(j\) occupies exactly \(1\) machine.
Motivated by many practical applications, such as high performance computing~\cite{drozdowski2009scheduling}, where parallelism is exploited to speed up the execution of jobs, many extensions of this problem have been studied.
One of which is scheduling \emph{moldable} jobs.
This is a natural extension where jobs can be executed on a variable number of machines.
There are many practical motivations for this problem formulation, as discussed in~\cite{bleuse2017scheduling,drozdowski2009scheduling,DBLP:conf/jsspp/FeitelsonR96}.

In moldable job scheduling we are given \(J = \{1, \dots, n\}\) jobs and \(m\) identical machines.
Each job can be assigned to \(k \leq m\) machines and the processing time of any job \(j\) is dependent on the number of assigned machines.
We denote the processing time of any job \(j\) as \(t(j, k)\).
Further, the \emph{work} of any job $j$ is defined as $w(j, k) = t(j, k) \cdot k$ and can be described as its area.
For \emph{monotone} moldable job scheduling we assume that for any job, the work function and time function are non-decreasing and non-increasing, respectively.
\begin{assumption}[monotony \cite{mounie2007}]
    Given any job $j$, for all $k \leq k' \leq m$, we have:
    \begin{enumerate}[(i)]
        \item \(w(j, k) \leq w(j, k')\)
        \item \(t(j, k) \geq t(j, k')\)
    \end{enumerate}
\end{assumption}

A solution to an instance of this problem is given by a \emph{schedule}, containing a start time $s_j$ for each job \(j \in J\), and an allotment \(\alpha(j) \in \{1, \dots, m\}\) for each job. A schedule is called \emph{feasible} if the following conditions are met:
\begin{itemize}
    \item A job starts its execution simultaneously on all its assigned machines
    \item A job may not be interrupted during its execution time
    \item Each machine can execute at most one job at a time
\end{itemize}
In the \emph{contiguous} variant of the problem, we additionally require each job to be processed on an adjacent set of machines.
In many practical applications obtaining such a schedule, where jobs are processed on adjacent machines, is beneficial as it allows for better memory locality and reduced communication overhead between machines \cite{bleuse2017scheduling}.
There are many other applications of this model, as discussed in \cite{DBLP:journals/jors/BlazewiczCMO11,DBLP:journals/eor/DelormeDKK19,DBLP:journals/eor/DolguiKKMS18,DBLP:journals/ior/FotakisMP25,mounie2007,WuZC23}.
In this paper we consider the problem of minimizing the makespan, i.e., the maximum completion time of any job in a schedule.

\subsection{Related Work}
Note that in moldable job scheduling the number of machines assigned to a job is fixed during the execution of the job.
Some older publications refer to this problem as \emph{malleable} job scheduling \cite{decker200654,DBLP:journals/scheduling/FotakisMP23,DBLP:journals/ior/FotakisMP25,jansen2002linear,jansen2008approximation,ludwig1994,mounie2007}.
However, the notation has recently changed to moldable job scheduling, whereas malleable job scheduling is now used to refer to the problem, where the number of machines assigned to a job can change during its execution.

Moldable job scheduling as well as the monotone variant are known to be NP-hard~\cite{du1989complexity, land2013boundingmoldable}. Turek, Wolf, and Yu~\cite{turek1992approximate} presented a \(2\)-approximation algorithm for the general case of moldable job scheduling. Later, Ludwig and Tiwari \cite{ludwig1994} presented an algorithm with the same approximation ratio but improved the running time to be polynomial in \(\log m\).
It has been proven that for moldable job scheduling (without monotony!) no polynomial-time approximation algorithm with a guarantee below \(3/2\) exists, unless P~=~NP~\cite{drozdowski1995no32approx,DBLP:journals/scheduling/Johannes06}.
But there exists a pseudo-polynomial \((\frac{5}{4}+\varepsilon)\)-approximation algorithm for scheduling contiguous moldable jobs with a time complexity of \(O(n\log n) \cdot m^{f(1/\varepsilon)}\), where \(f(\cdot)\) is a computable function~\cite{DBLP:conf/esa/JansenR19}.
For special cases of the problem there exist better approximation algorithms:
Jansen and Porkolab~\cite{jansen2002linear} presented an approximation scheme with linear running time for a constant number of machines \(m\); Decker, Lücking, and Monien~\cite{decker200654} gave a \(1.25\)-approximation under the additional assumption of identical jobs. Jobs are called identical if the execution time on any number of machines is the same for all jobs.

In this work we study the specific problem of \emph{monotone} moldable job scheduling.
Mounié, Rapine, and Trystram~\cite{DBLP:conf/spaa/MounieRT99,mounie2007} presented a \((\frac{3}{2} + \varepsilon)\)-approximation algorithm with a time complexity of \(O(nm \log \frac{1}{\epsilon})\) for the contiguous case.
Over the past years, various authors have improved this result.
Jansen and Thöle~\cite{DBLP:conf/icalp/JansenT08,jansen2008approximation} have shown the existence of an algorithm with an approximation ratio arbitrarily close to \(1.5\) for the contiguous problem and presented a \ac{PTAS} for the same problem without the contiguous restriction when \(m\) is bounded by a polynomial in \(n\).
Jansen and Land~\cite{jansenland2018schedulingmoldable} generalized the \ac{PTAS} in~\cite{jansen2008approximation} by presenting a \ac{FPTAS} for the contiguous case under the assumption of \(m > 8\frac{n}{\varepsilon}\).
Additionally, they improved the \((\frac{3}{2} + \varepsilon)\)-approximation algorithm in~\cite{mounie2007} to have a running time with only poly-logarithmic dependence on the number of machines \(m\).
These result were further improved by Grage, Jansen, and Ohnesorge~\cite{grage2023improved}.
Just recently, Wu, Zhang, and Chen~\cite{WuZC23} gave a \(\frac{3}{2}\)-approximation algorithm with a running time of \(O(nm \log nm)\) and stated the open question of whether an approximation algorithm for this problem exists with an approximation ratio better than \(\frac{3}{2}\) and a similar running time.

\renewcommand{\arraystretch}{1.2} 
\begin{table}[htbp]
    \centering
    \begin{tabular}{|l | l | l |}
        \hline
        \textbf{Ratio}             & \textbf{Complexity}                                                                                                               & \textbf{Author (Year)}                                      \\ [0.5ex]
        \hline\hline
        \(\frac{3}{2} + \epsilon\) & \(O(\log(1/\varepsilon) nm)\)                                                                                                     & Mounié, Rapine, Trystram (2007) \cite{mounie2007}           \\ [1ex]
        \hline

        \(\frac{3}{2} + \epsilon\) & \(O(\frac{n}{\varepsilon^2}\log m(\frac{\log m}{\epsilon} + \log^3(\varepsilon m)))\)                                             & Jansen, Land (2018) \cite{jansenland2018schedulingmoldable} \\[1ex]
        \hline

        \(\frac{3}{2} + \epsilon\) & \(O(n\log^2(\frac{1 + \log(\varepsilon m)}{\varepsilon}) + \frac{n}{\varepsilon}\log(\frac{1}{\varepsilon})\log(\varepsilon m))\) & Grage, Jansen, Ohnesorge (2023) \cite{grage2023improved}    \\[1ex]
        \hline
        \(\frac{3}{2}\)            & \(O(nm \log(nm))\)                                                                                                                & Wu, Zhang, Chen (2023) \cite{WuZC23}                        \\ [1ex]
        \hline
        \(1.4593 + \epsilon\)      & \(O(\log(1 / \varepsilon)nm)\)                                                                                                    & This Work                                                   \\ [1ex]
        \hline
    \end{tabular}
    \caption{Overview of approximation algorithms for monotone moldable job scheduling polynomial in \(n\), \(m\) and \(1/\varepsilon\).}
    \label{table:related work}
\end{table}
\subsection{Our Contribution}
The best known approximation ratio for the problem of monotone moldable job scheduling (non-contiguous) is \(1+\varepsilon\), given by the \ac{PTAS} in \cite{DBLP:conf/icalp/JansenT08,jansen2008approximation}.
This \ac{PTAS} uses a dynamic program to place jobs across \(\delta^{-2}\) many starting points.
When setting \(\varepsilon \coloneqq 0.5\), we get the number of starting points by calculating \(\delta \coloneqq \sigma_{36}\), such that \(\sigma_1 \coloneqq \frac{1}{18}\) and \(\sigma_k \coloneqq \frac{\sigma_{k-1}^3}{(4\cdot 7^2)}\).
This resolves to roughly \(\sigma_{36} \approx 10^{-1.2 \cdot 10^{17}}\).
The overall dynamic program then has a running time of roughly \(nm^{2 \cdot 10^{2.4\cdot 10^{17}}}\).
When comparing the \((\frac{5}{4} + \varepsilon)\)-approximation in \cite{DBLP:conf/esa/JansenR19}, to current \(\frac{3}{2}\)-approximation algorithms by setting \(\varepsilon \coloneqq \frac{1}{4}\), we get a running time of \(\Omega((mn)^{16^{4^{13}}})\).
While these algorithms are of theoretical interest, they have an impractically large running time.
Motivated by this, we focus on \emph{practically efficient} approximation algorithms (\cref{table:related work}).

Previously it was assumed that there exists no practically efficient \((\frac{3}{2} - \varepsilon)\)-approximation for the problem of monotone moldable job scheduling, for any \(\varepsilon > 0\) \cite{WuZC23}.
We break this barrier by presenting an approximation algorithm with a worst-case approximation ratio of \(\approx 1.4593 < 1.5\).
It is worth noting that the \ac{PTAS} presented in \cite{jansen2008approximation} cannot solve the contiguous variant of the problem, and it is unknown whether there exists a \ac{PTAS} for this variant of the problem.
Surprisingly, we show that our algorithm creates a contiguous schedule with a makespan of at most \(1.4593 \cdot \widetilde{OPT}\), where \(\widetilde{OPT}\) is a lower bound on the optimum of the non-contiguous problem.
This result also bounds the gap between the contiguous and non-contiguous variant of the problem.

Our algorithm has a running time of \(O(mn \log \frac{1}{\varepsilon})\).
This is technically only pseudo-polynomial in the encoding length of the problem because of the linear dependency on \(m\).
Similar to the algorithms in \cite{DBLP:conf/esa/JansenR19,jansen2008approximation}, we can use the \ac{FPTAS} presented in \cite{grage2023improved,jansenland2018schedulingmoldable} for instances with \(m > 8\frac{n}{\varepsilon}\).
Thus, we may assume that \(m \leq 8\frac{n}{\varepsilon}\).
Our running time of \(O(mn \log \frac{1}{\varepsilon})\) is fully polynomial in this case.
\begin{theorem}
    \label{thm:main}
    Let \(\varepsilon > 0\). For contiguous monotone moldable job scheduling, there exists an algorithm with an approximation ratio of \((1.4593 + \varepsilon)\) and a running time of \(O(nm \log \frac{1}{\varepsilon})\).
\end{theorem}

We achieve this result by improving the long-standing algorithm of Mounié, Rapine, and Trystram \cite{mounie2007}.
Although this algorithm has been improved many times with respect to the running time \cite{jansenland2018schedulingmoldable,grage2023improved}, improving the approximation ratio seemed to be a difficult task (except for eliminating the \(\varepsilon\) factor \cite{WuZC23}).
We manage this with the following new techniques:
(1) relaxing the standard 0/1-Knapsack used in their algorithm to a \acl{MCKP}, (2) packing the jobs into more containers, which results in more complex repair algorithms, and (3) analyzing a single critical case carefully with an elegant geometric argument.
Our method of analyzing this critical case results in the so-called Lambert W function appearing in the approximation ratio.

We believe that this result is of theoretical interest as it breaks the long-standing barrier of \(\frac{3}{2}\) for these practically efficient algorithms and the techniques employed, particularly the geometric argument, are potentially of broader interest.
A similar geometric analysis could be beneficial in other scheduling and packing problems~\cite{DBLP:conf/soda/EberleHRW25,DBLP:conf/focs/GalvezGHI0W17,DBLP:conf/wads/HarrenJPS11}.
Furthermore, our theoretical results are supported by practical experiments demonstrating that a schedule length of at most \(\frac{10}{7}OPT\) can be guaranteed in most cases.
This makes this algorithm a promising candidate for practical applications.

\subsection{Preliminaries}

As many of the previous works, we utilize a dual approximation framework as proposed by Hochbaum and Shmoys~\cite{hochbaum1987}. First, we obtain a lower and upper bound on the optimum makespan by a constant approximation algorithm \cite{turek1992approximate,ludwig1994} and then perform binary search in this interval to find the optimal makespan up to an accuracy of \(\varepsilon\).
Therefore, for this work we assume that we are given a makespan guess \(d\) and
want to find a schedule with makespan at most \(\lambda d\) or reject \(d\) if \(d < OPT\).

Similar to Mounié, Rapine, and Trystram \cite{mounie2007}, we define the \emph{canonical number of machines} as follows:
\begin{definition}[Canonical Number of Machines]
    Given a real number $h$, we define for each job $j$ its canonical number of machines $\gamma(j, h)$ as the minimal number of machines needed to execute $j$ in time at most $h$. If $j$ cannot be processed in time at most $h$ on $m$ machines, we set $\gamma(j, h) = \infty$ by convention.
\end{definition}

We do note that, due to monotony, this number can be found in time $O(\log m)$ using binary search.

\section{Description of the Algorithm}
\label{sec:algorithm}

In order to motivate the core idea of our algorithm, we first give a brief overview of the algorithm presented by Mounié, Rapine, and Trystram \cite{mounie2007}:
First, find a partition \(\mathcal{S}_1 \sqcup \mathcal{S}_2 = J_B = \set{j \in J}[t(j, 1) > \frac{d}{2}]\) with some specific properties.
Then, schedule the first set \(\mathcal{S}_1\) with a total width and height of \(m\) and \(d\), respectively. The second set \(\mathcal{S}_2\) is scheduled with a height of at most \(\frac{d}{2}\) on top, which then results in a total height of at most \(\frac{3}{2}d\). Afterwards, the remaining small jobs \(J_S = \set{j \in J}[t(j, 1)\leq \frac{d}{2}]\) are scheduled greedily.
For ease of notation we will, similar to~\cite{mounie2007,jansen2002linear,grage2023improved}, call the sets \(\mathcal{S}_1, \mathcal{S}_2\) \emph{shelf 1} and \emph{shelf 2}, respectively.
We describe the maximum allowed execution time of jobs in each shelf as the \emph{height} of the shelf, e.g. in the algorithm of~\cite{mounie2007} the height of shelf 1 is \(d\) and the height of shelf 2 is \(\frac{d}{2}\).

Although this algorithm is quite tailored to an approximation ratio of \(\frac{3}{2}\), we manage to adapt it to obtain a better approximation ratio.
A key insight is that we can partition the jobs into more sets to more closely resemble the allotment in an optimal schedule.
We believe that this technique might be of more general interest.
Improving the approximation ratio of this algorithm comes with a few major challenges:
(1) Greedily adding jobs \(j \in J_S\) with \(t(j, 1) \leq \frac{d}{2}\) will result in an approximation ratio \(\geq 3/2\).
(2) We need to be more careful about the partitioning of jobs into shelves.

We overcome the first challenge by defining \(J_S = \set{j \in J}[t(j, 1)\leq\frac{3}{7}d]\) as the set of small jobs, and \(J_B = J \setminus J_S\) as the set of big jobs.
This choice, while it may seem arbitrary, will become clear in the subsequent analysis.
The minimal work of the small jobs in any schedule is denoted as \(\mathcal{W}_S = \sum_{j \in J_S} w(j, 1)\).

For the partition we allow a third shelf, which is allowed to exceed the width of \(m\). By exploiting the work monotony of jobs, we can later reduce the number of machines assigned to jobs in this shelf.
This leads us to the central theorem of our work:

\begin{theorem}
    \label{thm:dual-approx}
    Given a partitioning of jobs $\mathcal{C}_1 \sqcup \mathcal{C}_2 \sqcup \mathcal{C}_3 = J_B$ s.t.
    \renewcommand{\labelenumi}{C\arabic{enumi}}
    \renewcommand{\theenumi}{C\arabic{enumi}}
    \begin{enumerate}
        \item \label{const1:work}
              $\sum_{j\in \mathcal{C}_1}w(j, \gamma(j, d))
                  + \sum_{j\in \mathcal{C}_2}w\left(j, \gamma\left(j, \frac{4}{7}d\right)\right)
                  + \sum_{j\in \mathcal{C}_3}w\left(j, \gamma\left(j, \frac{3}{7}d\right)\right)
                  \leq md - \mathcal{W}_S$

        \item \label{const2:machines}
              $\sum_{j\in \mathcal{C}_1}{\gamma(j, d)} + \sum_{j \in \mathcal{C}_2}\gamma\left(j, \frac{4}{7}d\right)\cdot \frac{1}{2}\leq m$ \mbox{}\hfill
    \end{enumerate}
    one can compute a contiguous schedule with makespan at most $\lambda d$ in time $O(mn)$, for any \(\lambda > -\frac{1}{3}W_{-1}(-\frac{3}{e^4}) \approx 1.4593\).
\end{theorem}

The Lambert W function, also called omega function or product logarithm is defined such that \(ye^y = x\) holds if and only if \(y = W(x)\). The \(W_{-1}\) branch is used for \(-\frac{1}{e}\leq x < 0\) \cite{corless1996Lambert}.
The fact that the Lambert W function appears is due to the nature of our analysis, where we define integrals to get a bound on the work of jobs in a schedule.
This technique, while not new, is not commonly used in the analysis of scheduling problems.
Fotakis, Matuschke, and Papadigenopoulos \cite{DBLP:journals/scheduling/FotakisMP23} use a similar analysis technique, but the Lambert W function does not appear in their analysis.

\subsection{Finding a Partition}
We start by showing the existence of a partition of jobs into three classes, as required by \cref{thm:dual-approx}, given a schedule with makespan at most \(d\) exists.

\begin{lemma}
    \label{lem:KP_W}
    Given any schedule with makespan at most \(d\), there exists a partition \(\mathcal{C}_1 \sqcup \mathcal{C}_2 \sqcup \mathcal{C}_3 = J_B\) that satisfies constraints \ref{const1:work} and \ref{const2:machines} of \cref{thm:dual-approx}.
\end{lemma}
\begin{proof}
    Consider any schedule for an instance \(I\) with makespan \(d^* \leq d\). The total work of all small jobs, defined as \(J_S = \set{j \in J}[t(j, 1) \leq \frac{3}{7}d]\), is at least \(\mathcal{W}_S = \sum_{j \in J_S} t(j, 1)\), due to the assumption of work monotony. Therefore, the work of all remaining jobs is at most \(\mathcal{W} = md - \mathcal{W}_S\). For any job \(j \in J_B\), we will denote the number of machines assigned to this job in the optimal schedule as \(OPT(j)\).
    We now argue that the partition \(\mathcal{C}_1 \coloneqq \set{j \in J_B}[t(j, OPT(j)) > \frac{4}{7}d]\), \(\mathcal{C}_2 \coloneqq \set{j \in J_B}[\frac{4}{7}d \geq t(j, OPT(j)) > \frac{3}{7}d]\), and \(\mathcal{C}_3 \coloneqq \set{j \in J_B}[\frac{3}{7}d \geq t(j, OPT(j))]\) of jobs in \(J_B\) satisfies both constraints \ref{const1:work} and \ref{const2:machines}.

    Due to monotony, we know that the work of any job in \(\mathcal{C}_3\) is at least \(w(j, \gamma(j, \frac{3}{7}d))\). Therefore:
    \begin{align}
        \label{eq:lemKPW:5} \sum_{j \in \mathcal{C}_3} w\left(j, \gamma\left(j, \frac{3}{7}d\right)\right) & \leq \sum_{j \in \mathcal{C}_3} w(j, OPT(j))
    \end{align}

    Each job in \(\mathcal{C}_1\) uses at least \(\gamma(j, d)\) machines and has a work of at least \(w(j, \gamma(j, d))\), therefore:
    \begin{align}
        \label{eq:lemKPW:1} \sum_{j \in \mathcal{C}_1} w(j, \gamma(j, d)) & \leq \sum_{j \in \mathcal{C}_1} w(j, OPT(j)) \\
        \label{eq:lemKPW:2} \sum_{j \in \mathcal{C}_1} \gamma(j, d)       & \leq \sum_{j \in \mathcal{C}_1} OPT(j)
    \end{align}

    For the last partition \(\mathcal{C}_2\), first observe that in any optimal schedule of height at most \(d\) no job from \(\mathcal{C}_2\) can be scheduled on the same machine as a job from \(\mathcal{C}_1\).
    Thus, all jobs in \(\mathcal{C}_2\) are scheduled on at most \(m - \sum_{j \in \mathcal{C}_1} OPT(j)\) machines. Each job \(j \in \mathcal{C}_2\) uses at least \(\gamma(j, \frac{4}{7}d)\) machines and, since \(t(j, OPT(j)) > \frac{3}{7}d\), at most two such jobs can be scheduled on the same machine. Therefore:
    \begin{align}
        \label{eq:lemKPW:4} \sum_{j \in \mathcal{C}_2} \gamma\left(j, \frac{4}{7}d\right) \cdot \frac{1}{2} & \leq m - \sum_{j \in \mathcal{C}_1} OPT(j)
    \end{align}
    For the work area, we give a similar argument as above and get:
    \begin{align}
        \label{eq:lemKPW:3} \sum_{j \in \mathcal{C}_2} w\left(j, \gamma\left(j, \frac{4}{7}d\right)\right) & \leq \sum_{j \in \mathcal{C}_2} w(j, OPT(j))
    \end{align}

    Combining \Cref{eq:lemKPW:2,eq:lemKPW:4}, yields constraint \ref{const2:machines} and combining \Cref{eq:lemKPW:1,eq:lemKPW:3,eq:lemKPW:5}, yields constraint \ref{const1:work}.
    This proves the lemma.
\end{proof}

Thus, given a makespan guess \(d \geq \textsc{OPT}\), we know a partition of jobs \(\mathcal{C}_1 \sqcup \mathcal{C}_2 \sqcup \mathcal{C}_3 = J_B\) exists such that the constraints of \cref{thm:dual-approx} hold.
Finding such a partition is equivalent to solving a \ac{MCKP}.

\begin{definition}[\acf{MCKP}]
    Given a capacity \(m \in \mathbb{N}\), \(k \in \mathbb{N}\) classes and \(n \in \mathbb{N}\) items, each item \(j \in [n]\) with a cost \(c_{jl}\) and a size \(s_{jl}\) for each \(l \in [k]\).
    Find a partition \(\mathcal{C}_1 \sqcup \dots \sqcup \mathcal{C}_k = [n]\) with:
    \begin{align*}
         & \min         \sum_{l = 1}^{k}{\sum_{j \in \mathcal{C}_l}{c_{jl}}}        \\
         & \text{s.t.}  \sum_{l = 1}^{k}{\sum_{j \in \mathcal{C}_l}{s_{jl}}} \leq m
    \end{align*}
\end{definition}

Solving the \ac{MCKP} can be done in time \(O(mn)\) using dynamic programming \cite{knapsackproblems}.
Given a makespan guess \(d \geq OPT\), finding a partition that satisfies the constraints of \cref{thm:dual-approx} is equivalent to solving a \ac{MCKP} with capacity \(m\) and \(k = 3\) classes, where each job \(j \in J_B\) is an item with costs \(c_{j1} = w(j, \gamma(j, d))\), \(c_{j2} = w(j, \gamma(j, \frac{4}{7}d))\), and \(c_{j3} = w(j, \gamma(j, \frac{3}{7}d))\) and sizes \(s_{j1} = \gamma(j, d)\), \(s_{j2} = \frac{1}{2}\gamma(j, \frac{4}{7}d)\), and \(s_{j3} = 0\).
Therefore, we will assume for the remainder of this work that we are given a partition \(\mathcal{C}_1 \sqcup \mathcal{C}_2 \sqcup \mathcal{C}_3 = J_B\) that satisfies the constraints of \cref{thm:dual-approx}.

\subsection{Constructing a Schedule}

In this section we assume that we are given a makespan guess \(d \geq OPT\) and a partition \(\mathcal{C}_1 \sqcup \mathcal{C}_2 \sqcup \mathcal{C}_3 = J_B\) that satisfies the constraints of \cref{thm:dual-approx}.
We will prove \cref{thm:dual-approx} by providing an algorithm with the required approximation ratio.
An overview of this algorithm is given in \cref{alg:DualApprox}.


\begin{algorithm}[H]
    \caption{Description of the Dual Approximation Algorithm}\label{alg:DualApprox}
    \begin{algorithmic}[1]
        \Require \(m\) identical machines, \(n\) moldable jobs, and a makespan guess \(d\)
        \State Let \(J_S = \{j \in J | t(j, 1) \leq \frac{3}{7}d\}\) and \(J_B = J \setminus J_S\).
        \State Construct partition \(\mathcal{C}_1 \sqcup \mathcal{C}_2 \sqcup \mathcal{C}_3 = J_B\) via \ac{MCKP}. \Comment{\cref{lem:KP_W}}
        \If{\(W(\mathcal{C}_1, \mathcal{C}_2, \mathcal{C}_3) > dm - \mathcal{W}_S\)}
        \State \textbf{reject} \(d\).
        \EndIf
        \State Distribute jobs in \(\mathcal{C}_1, \mathcal{C}_2, \mathcal{C}_3\) onto shelves according to \Cref{lem:machinesS0S1}. \Comment{\cref{def:3shelf-schedule}}
        \State Apply \Cref{alg:repair1}.
        \If{\(q < m/6\)}
        \State Apply \Cref{alg:RepairS2}. \Comment{\cref{lem:RepairS2-1:correctness}}
        \Else
        \State Apply \Cref{alg:RepairS2-2}. \Comment{\cref{lem:RepairS2-2:correctness}}
        \EndIf
        \State Greedily add small jobs \(J_S\) via \Cref{alg:addSmall}. \Comment{\cref{lem:addsmall}}
    \end{algorithmic}
\end{algorithm}

Our goal is now to construct an intermediate schedule, as shown in \Cref{fig:3shelf-schedule-rules} (left) for all jobs \(J_B = J \setminus J_S\).
Any job $j \in \mathcal{C}_1$ is scheduled on $\gamma(j, d)$ machines at time 0, any job $j \in \mathcal{C}_2$ is scheduled on $\gamma(j, \frac{4}{7}d)$ machines at time 0, and any job \(j \in \mathcal{C}_3\) on \(\gamma(j, \frac{3}{7}d)\) machines such that it completes at time \(\frac{10}{7}d\).
When analyzing the resulting schedule, we notice that (1) the jobs in \(\mathcal{C}_1\) and \(\mathcal{C}_2\) require more than \(m\) machines, and (2) the jobs in \(\mathcal{C}_3\) might also require more than \(m\) machines.

\begin{figure}[h]
    \centering
    \resizebox{\columnwidth}{!}{%
        \centering
\begin{tikzpicture}
    \draw (0, 0) -- (10, 0);
    \draw[<->] (0, -0.2) -- (10, -0.2)
    node[midway, anchor=north] {\(m\)};
    \draw (0, 0) -- (0, 7);
    \node[anchor=east] (D)    at (0, 5) {$d$};
    \node[anchor=east] (10/7D) at (0, 7) {$\frac{10}{7}d$};
    \node[anchor=east] (4/7D) at (0, 3) {$\frac{4}{7}d$};
    \draw[dashed] (4/7D) -- ($(10, 0)!(4/7D)!(10, 0)$);
    \draw[dashed] (D) -- ($(10, 0)!(D)!(10, 0)$);
    \draw[dashed] (10/7D) -- ($(10, 0)!(10/7D)!(10, 0)$);

    \begin{scope}[fill=blue, fill opacity=0.7, text=blue]
        \filldraw (0, 0) rectangle +(-3, 2.4);
        \filldraw (0, 0) rectangle +(2, 2.8);
        \filldraw (2, 0) rectangle +(1, 3);
    \end{scope}
    \begin{scope}[fill=red, fill opacity=0.7, text=red]
        \filldraw (3, 0) rectangle +(1, 3.5);
        \filldraw (4, 0) rectangle +(2, 4.7);
        \filldraw (6, 0) rectangle +(4, 4.8);
    \end{scope}
    \begin{scope}[yshift=7cm, fill=yellow, fill opacity=0.7, text=yellow]
        \filldraw (0, 0) rectangle +(5, -1.7);
        \filldraw (5, 0) rectangle +(4.5, -1.5);

    \end{scope}
\end{tikzpicture}
\hspace{1cm} 

\begin{tikzpicture}
    \draw (0, 0) -- (10, 0);
    \draw[<->] (0, -0.2) -- (10, -0.2)
    node[midway, anchor=north] {\(m\)};
    \draw (0, 0) -- (0, 7);
    \node[anchor=east] (D)    at (0, 5) {$d$};
    \node[anchor=east] (10/7D) at (0, 7) {$\frac{10}{7}d$};
    \node[anchor=east] (4/7D) at (0, 3) {$\frac{4}{7}d$};
    \draw[dashed] (4/7D) -- ($(10, 0)!(4/7D)!(10, 0)$);
    \draw[dashed] (D) -- ($(10, 0)!(D)!(10, 0)$);
    \draw[dashed] (10/7D) -- ($(10, 0)!(10/7D)!(10, 0)$);

    \begin{scope}[fill=blue, fill opacity=0.7, text=blue]
        \filldraw (0, 0) rectangle +(1, 5.5);
        \filldraw (1, 3.8) rectangle +(1, 2.5);
        \filldraw (1, 0) rectangle +(2, 3.8);
    \end{scope}
    \begin{scope}[fill=red, fill opacity=0.7, text=red]
        \filldraw (3, 0) rectangle +(1, 3.5);
        \filldraw (4, 0) rectangle +(2, 4.7);
        \filldraw (6, 0) rectangle +(4, 4.8);
    \end{scope}      
    \begin{scope}[yshift=7cm, fill=yellow, fill opacity=0.7, text=yellow]
        \filldraw (2, 0) rectangle +(5, -1.7);
        \filldraw (7, 0) rectangle +(5.5, -1.5);
    \end{scope}

\end{tikzpicture}%
    }
    \caption{Schedule before \cref{lem:machinesS0S1} (left) and after (right). Pink jobs are in \(\mathcal{C}_1\), blue jobs in \(\mathcal{C}_2\), and yellow jobs in \(\mathcal{C}_3\).}
    \label{fig:3shelf-schedule-rules}
\end{figure}
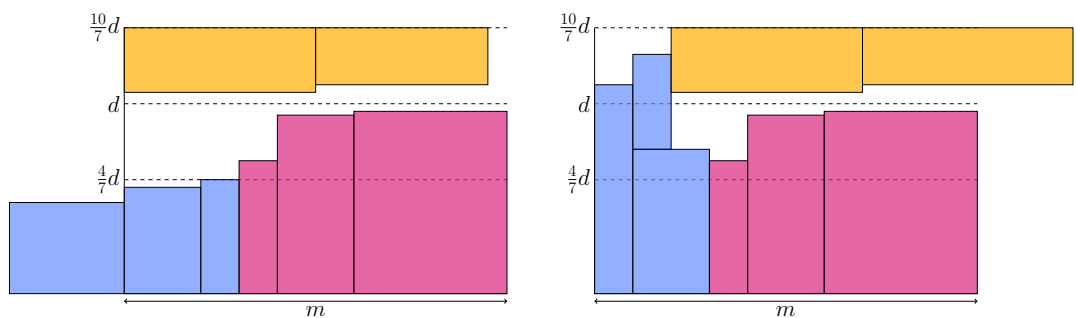

We fix the first problem by scheduling the jobs \(j \in \mathcal{C}_2\) on fewer machines and focus on the second problem in \cref{sec:repair}.
Note that all jobs are scheduled in a way such that Constraint \ref{const1:work} remains satisfied and the total work of all big jobs is at most \(md - \mathcal{W}_S\). This remains true, due to work-monotony, for the remainder of this work by ensuring that no alterations to this schedule increase the number of machines assigned to any job.

\begin{lemma}
    \label{lem:machinesS0S1}
    All jobs in \(\mathcal{C}_2\) can be scheduled on \(m - \sum_{j \in \mathcal{C}_1}{\gamma(j, d)}\) machines without exceeding the makespan of \(\frac{10}{7}d\).
\end{lemma}
\begin{proof}
    First notice that by Constraint \ref{const2:machines}, we get \(\sum_{j \in \mathcal{C}_2}{\gamma(j, \frac{4}{7}d) \cdot \frac{1}{2}} \leq m - \sum_{j \in \mathcal{C}_1}\gamma(j, d)\).
    Therefore, we need to effectively schedule each job, such that it requires at most \(\gamma(j, \frac{4}{7}d) \cdot \frac{1}{2}\) machines.

    We analyze jobs \(j \in \mathcal{C}_2\) based on their canonical number of machines \(\gamma(j, \frac{4}{7}d)\).
    \begin{case}[\(\gamma(j, \frac{4}{7}d) \geq 4\)]
        These jobs can be scheduled in such a way that they do not exceed an execution time of \(\frac{10}{7}d\) when assigned to \(\lfloor \gamma(j, \frac{4}{7}d) / 2 \rfloor\) machines. Due to work monotony, we have:
        \begin{align*}
            w\left(j, \gamma\left(j, \frac{4}{7}d\right)\right)                                                                         & \geq w\left(j, \left\lfloor \frac{\gamma\left(j, \frac{4}{7}d\right)}{2} \right\rfloor\right)                                                         \\
            \Leftrightarrow t\left(j, \gamma\left(j, \frac{4}{7}d\right)\right) \cdot \gamma\left(j, \frac{4}{7}d\right)                & \geq t\left(j, \left\lfloor \frac{\gamma(j, \frac{4}{7}d)}{2} \right\rfloor\right) \cdot \left\lfloor \frac{\gamma(j, \frac{4}{7}d)}{2} \right\rfloor \\
            \Rightarrow \frac{4}{7}d \cdot \frac{\gamma(j, \frac{4}{7}d)}{\left\lfloor \frac{\gamma(j, \frac{4}{7}d)}{2} \right\rfloor} & \geq t\left(j, \left\lfloor \frac{\gamma(j, \frac{4}{7}d)}{2} \right\rfloor\right)
        \end{align*}
        This implies that \(t\left(j, \left\lfloor \gamma(j, \frac{4}{7}d) / 2 \right\rfloor\right)\) is smaller than or equal to \(\frac{10}{7}d\) for any \(\gamma(j, \frac{4}{7}d) \geq 4\).
    \end{case}

    \begin{case}[\(\gamma(j, \frac{4}{7}d) = 2\)]
        These jobs can be scheduled on 1 machine without exceeding the makespan of \(\frac{10}{7}d\), exactly halving the number of required machines.
    \end{case}

    \begin{case}[\(\gamma(j, \frac{4}{7}d) = 3\)]
        If there are multiple such jobs, they can be scheduled in pairs on top of each other without idle time. In the remainder of this work, we may treat these jobs as a single bigger job. Thus, we may assume there exists at most one of these jobs and denote it as \(j_3\).
    \end{case}

    \begin{case}[\(\gamma(j, \frac{4}{7}d) = 1\)]
        These jobs can also be scheduled in pairs on top of each other. In the remainder of this work, we may treat these jobs as a single bigger job. A remaining single job, denoted \(j_1\), may need to be processed differently.
    \end{case}

    If neither \(j_1\) nor \(j_3\) exists, we have successfully scheduled all jobs in \(\mathcal{C}_2\) such that they require no more than \(m - \sum_{j \in \mathcal{C}_1}\gamma(j, d)\) machines.

    If only \(j_1\) exists, the schedule still fits within the machine limit since \(m - \sum_{j \in \mathcal{C}_1}\gamma(j, d)\) is of integer value.

    If only \(j_3\) exists, it can be scheduled on \(\gamma(j_3, \frac{10}{7}d) \leq 2\) machines, and again the schedule fits within the machine limit.

    If both \(j_1\) and \(j_3\) exist, \(j_3\) is scheduled on 2 machines, with a height of at most \(\frac{4}{7}d \cdot \frac{3}{2} = \frac{6}{7}d\), and \(j_1\) is placed on top of one of these machines.
    We consider these two jobs as one job of height \(t(j_3, 2) + t(j_1, 1)\) scheduled on 1 machine, and one job of height \(t(j_3, 2)\) scheduled on 1 machine.
\end{proof}
\begin{remark}
    \label{rem:noncontiguous}
    Note that splitting the job \(j_3\) into two parts might break the contiguous property of the schedule.
    However, we can simply reorder the machines such that this job is scheduled on adjacent machines.
    We will revisit this problem in each of the following proofs individually to ensure at any point that both parts of \(j_3\) can be scheduled on adjacent machines.
\end{remark}

For the following sections, we will refer to the schedule constructed in \cref{lem:machinesS0S1} as a 3-shelf schedule (\cref{def:3shelf-schedule} with \(\lambda = \frac{10}{7}\)).
To that end, we distribute jobs in \(\mathcal{C}_1 \sqcup \mathcal{C}_2 \sqcup \mathcal{C}_3\) onto three shelves, as shown in \Cref{fig:3shelf-schedule-trans}. Jobs in \(\mathcal{C}_1 \sqcup \mathcal{C}_2\) are distributed onto the first two shelves (\(S_0\) and \(S_1\)) such that any job with an execution time greater than \(d\) is in \(S_0\) and the remaining jobs are in \(S_1\).
The jobs in \(\mathcal{C}_3\) are put into \(S_2\) with an execution time of at most \(\frac{3}{7}d\).

\begin{definition}[3-Shelf Schedule]
    \label{def:3shelf-schedule}
    Given a makespan guess \(d\) and \(\lambda \in [\frac{10}{7}, \frac{3}{2})\), a 3-shelf schedule is a partition of jobs \(S_0 \sqcup S_1 \sqcup S_2 = J_B = J \setminus J_S\) such that:
    \begin{itemize}
        \item \(\sum_{j \in S_0}{w(j, \gamma(j, \lambda d))} + \sum_{j \in S_1}{w(j, \gamma(j, d))} + \sum_{j \in S_2}{w(j, \gamma(j, (\lambda - 1)d))} \leq md - \mathcal{W}_S\)
        \item \(\sum_{j \in S_0}{\gamma(j, \lambda d)} + \sum_{j \in S_1}{\gamma(j, d)} \leq m\)
    \end{itemize}
\end{definition}

Note that by \cref{lem:machinesS0S1}, \(\sum_{j \in S_0}{\gamma(j, \frac{10}{7}d)}+\sum_{j \in S_1}{\gamma(j, d)} \leq m\).
The total work of the jobs is still bounded by \(md - \mathcal{W}_S\).
As illustrated in \Cref{fig:3shelf-schedule-trans}, we now have to consider the remaining problem that shelf 2 might require too many machines.

\subsection{Repairing the 3-Shelf Schedule}
\label{sec:repair}
This section is devoted to repairing the 3-shelf schedule obtained in the previous section. Although the above knapsack formulation suggests a makespan guarantee of \(\frac{10}{7}d\) is achievable, we will have to relax this to obtain a feasible schedule.
To this end, we parameterize the approximation ratio of our algorithm by a parameter \(\lambda \in [\frac{10}{7}, \frac{3}{2})\), to be chosen later.
Thus, we aim to construct a schedule of height at most \(\lambda d\), where we allow jobs in shelf 2 to have an execution time of at most \((\lambda - 1)d\).

We will denote the number of machines used by \(S_0\) and \(S_2\) as \(m_0 \coloneqq \sum_{j \in S_0}{\gamma(j, \lambda d)}\) and \(m_2 \coloneqq \sum_{j \in S_2}{\gamma(j, (\lambda - 1)d)}\), respectively.
By construction \(m - m_0\) machines are used to schedule jobs in \(S_1\) and \(S_2\).
\cref{lem:machinesS0S1} guarantees that the jobs in \(S_1\) require at most \(m - m_0\) machines with height of at most \(d\).
We denote the number of idle machines in shelf 1, i.e. machines that do not execute any job \(j \in S_1\), as \(q\).
Note that it is not yet guaranteed that the jobs in \(S_2\) can be scheduled on the remaining \(m - m_0\) machines with a maximal height of \((\lambda - 1)d\), for any \(\lambda \in [\frac{10}{7}, \frac{3}{2})\).
For the remainder of this section, we assume that we are given a 3-shelf schedule with \(m_2 > m - m_0\), as the schedule is feasible otherwise (\Cref{fig:3shelf-schedule-trans}).

The following proofs use a lower bound on the total work of jobs to show the desired approximation ratio.
Remember that, by Constraint \ref{const1:work} of \cref{thm:dual-approx}, the given schedule satisfies \(\mathcal{W}_0 + \mathcal{W}_1 + \mathcal{W}_2 \leq md\), where \(\mathcal{W}_k\) denotes the total work area of jobs in shelf \(k\).
Since, by construction, the total work of jobs in shelf 0 is greater than \(dm_0\), we get \(\mathcal{W}_1 + \mathcal{W}_2 \leq d(m - m_0)\).
For ease of notation, we will w.l.o.g ignore shelf 0 and assume that \(m_0 = 0\) as illustrated in \Cref{fig:3shelf-schedule-trans} (right).

\begin{restatable}{lemma}{LemmaRepairProperties}
    \label{lem:repair-properties}
    Given a 3-shelf schedule and any \(\lambda \in [\frac{10}{7}, \frac{3}{2})\), we can either find a feasible contiguous schedule with makespan at most \(\lambda d\) in time \(O(mn)\), or the following properties hold:
    \begin{enumerate}[(i)]
        \item At most one machine executes a job \(j \in S_1\) with height less than \(\frac{\lambda}{2}d\).\label{lem:repair-properties:1}
        \item The work area of any job in \(S_2\) is greater than \(\lambda dq\).\label{lem:repair-properties:2}
        \item The work area of \(S_1\) is greater than \(\frac{\lambda}{2}d\left(m - q\right)\).\label{lem:repair-properties:3}
    \end{enumerate}
\end{restatable}
\begin{proof}
    The proof of this lemma uses techniques originally proposed by Mounié, Rapine, and Trystram \cite{mounie2007}. We give a full proof in Appendix \ref{sec:app:repair2:properties} for completeness.
\end{proof}

\begin{figure}[h]
    \centering
    \resizebox{\columnwidth}{!}{%
        \begin{tikzpicture}
    \draw (0, 0) -- (10, 0);
    \draw[<->] (0, -0.2) -- (10, -0.2)
    node[midway, anchor=north] {\(m\)};
    \draw[<->] (0, -0.8) -- (2, -0.8)
    node[midway, anchor=north] {\(m_0\)};
    \draw[<->] (9, -0.8) -- (10, -0.8)
    node[midway, anchor=north] {\(q\)};

    \draw (0, 0) -- (0, 7);
    \node[anchor=east] (D)    at (0, 5) {$d$};
    \node[anchor=east] (10/7D) at (0, 7) {$\lambda d$};
    \draw[dashed] (D) -- ($(10, 0)!(D)!(10, 0)$);
    \draw[dashed] (10/7D) -- ($(10, 0)!(10/7D)!(10, 0)$);
    \draw[<->] (2, 7.2) -- (12.5, 7.2)
    node[midway, anchor=south] {\(m_2\)};

    \begin{scope}[fill=blue, fill opacity=0.7]
        \filldraw[] (0, 0) rectangle +(1, 5.5);
        \filldraw[] (1, 0) rectangle +(1, 6.3);
    \end{scope}
    \begin{scope}[fill=red, fill opacity=0.7]
        \filldraw[] (2, 0) rectangle +(1, 3.8);
        \filldraw[] (3, 0) rectangle +(1, 3.5);
        \filldraw[] (4, 0) rectangle +(2, 4.7);
        \filldraw[] (6, 0) rectangle +(3, 4.8);
    \end{scope}

    \begin{scope}[yshift=7cm, fill=yellow, fill opacity=0.7]
        \filldraw[] (2, 0) rectangle +(5, -1.7);
        \filldraw[] (7, 0) rectangle +(5.5, -1.5);
    \end{scope}

\end{tikzpicture}
\hspace{1cm} 

\begin{tikzpicture}
    \draw (0, 0) -- (8, 0);
    \draw[<->] (0, -0.2) -- (8, -0.2)
    node[midway, anchor=north] {\(m\)};
    \draw[<->] (7, -0.8) -- (8, -0.8)
    node[midway, anchor=north] {\(q\)};

    \draw (0, 0) -- (0, 7);
    \node[anchor=east] (D)    at (0, 5) {$d$};
    \node[anchor=east] (10/7D) at (0, 7) {$\lambda d$};
    \draw[dashed] (D) -- ($(8, 0)!(D)!(8, 0)$);
    \draw[dashed] (10/7D) -- ($(8, 0)!(10/7D)!(8, 0)$);
    \draw[<->] (0, 7.2) -- (10.5, 7.2)
    node[midway, anchor=south] {\(m_2\)};
    \begin{scope}[xshift=-2cm]

        \begin{scope}[fill=red, fill opacity=0.7]
            \filldraw[] (2, 0) rectangle +(1, 3);
            \filldraw[] (3, 0) rectangle +(1, 3.5);
            \filldraw[] (4, 0) rectangle +(2, 4.7);
            \filldraw[] (6, 0) rectangle +(3, 4.8);
        \end{scope}

        \begin{scope}[yshift=7cm, fill=yellow, fill opacity=0.7]
            \filldraw[] (2, 0) rectangle +(5, -1.7);
            \filldraw[] (7, 0) rectangle +(5.5, -1.5);
        \end{scope}
    \end{scope}
\end{tikzpicture}%
    }
    \caption{Infeasible 3-shelf schedule. Blue jobs are in \(S_0\), pink jobs in \(S_1\), and yellow jobs in \(S_2\).}
    \label{fig:3shelf-schedule-trans}
\end{figure}
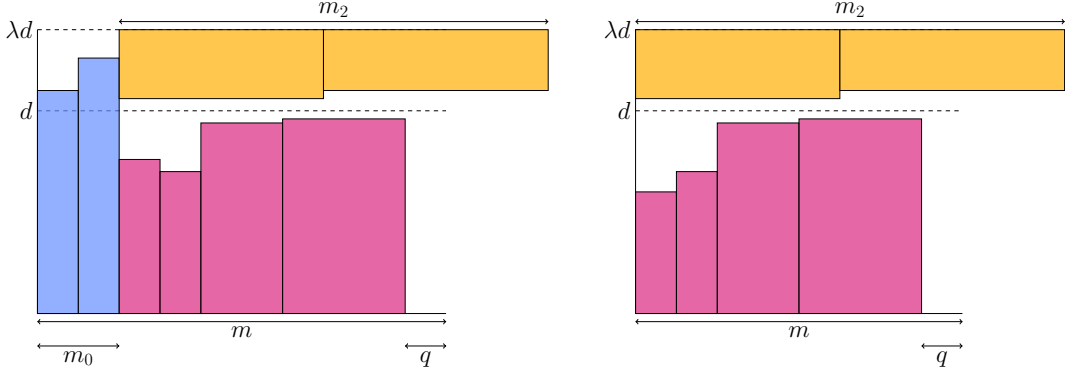

Our goal for the remainder of this section is now to show that \(S_2\) does not require more than \(m\) machines.
The new idea is to compress jobs in \(S_2\) to a height greater than the height of their shelf (\((\lambda - 1)d\)), since the jobs in \(S_1\) cannot be utilizing the entire height of their shelf.
We formulate multiple new repair algorithms that utilize this observation, and show that they can construct a feasible schedule in several ways.

\begin{restatable}{lemma}{LemmaRepairCorrectness}
    \label{lem:RepairS2-1:correctness}
    Given a 3-shelf schedule with \(q \leq m/6\), we can find a feasible contiguous schedule with makespan at most \(\frac{10}{7}d\) (if \(q = 0\)), and \(\frac{13}{9}d\) (if \(0 < q \leq m/6\)) in time \(O(mn)\).
\end{restatable}
\begin{proof}
    The proof of this lemma uses fairly simple techniques, therefore we defer it to Appendix \ref{sec:app:repair2-1:correctness}.
\end{proof}

It remains to show that we can also find a feasible schedule for the case \(q > m/6\).
We prepare this proof with the following observation:

\begin{observation}
    Given a 3-shelf schedule with \(q > m/6\) and \(\lambda \geq \frac{4}{3}\), the number of jobs in \(S_2\) is at most \(|S_2| \leq 1\).
\end{observation}
\begin{proof}
    Suppose \(|S_2| > 1\).
    By \cref{lem:repair-properties}.(\ref{lem:repair-properties:2}), we have (\(\mathcal{W}_2 > \lambda dq \cdot |S_2| \geq 2\lambda dq\)).
    With \cref{lem:repair-properties}.(\ref{lem:repair-properties:3}), we get:
    \begin{align*}
        \mathcal{W}_1 + \mathcal{W}_2                             & \leq dm                                 \\
        \Rightarrow \frac{\lambda}{2}d(m-q) + 2\lambda dq         & < dm                                    \\
        \Leftrightarrow \frac{\lambda}{2}m + \frac{3}{2}\lambda q & < m                                     \\
        \Leftrightarrow \frac{3}{4}\lambda m                      & < m           &  & \text{(\(q > m/6\))} \\
        \Leftrightarrow \lambda                                   & < \frac{4}{3}
    \end{align*}
    This is a contradiction to \(\lambda \geq \frac{4}{3}\).
\end{proof}

\begin{algorithm}[H]
    \caption{RepairS2-2}\label{alg:RepairS2-2}
    \begin{algorithmic}[1]
        \State Denote the single job in $S_2$ as $j_0$
        \State Sort all jobs in $S_1$ in descending order of execution time
        \State Set \(i = 0\)
        \For{\textbf{each} \(i \in \{0, \dots, m - q - 1\}\)}
        \State Place $j_0$ on the $m - i$ least loaded machines
        \If{\(C_{\max} \leq \lambda d\)}
        \State \textbf{return}
        \EndIf
        \EndFor
    \end{algorithmic}
\end{algorithm}

With this, we may assume that \(|S_2| = 1\) and denote the single job in \(S_2\) as \(j_0\).
This observation can be exploited by enumerating all possible allotments for this job (\Cref{alg:RepairS2-2}).

The following lemma analyzes the correctness of the last (and most complex) case.

\begin{lemma}
    \label{lem:RepairS2-2:correctness}
    Given a 3-shelf schedule with \(q > m/6\), we can find a feasible schedule with makespan at most $\lambda d$, for any \(\lambda \geq -\frac{1}{3}W_{-1}(-\frac{3}{e^4}) \approx 1.4593\).
\end{lemma}
\begin{proof}
    We intend to prove the lemma by showing that the height of any job in \(S_1\) scheduled on machine \(i \in \{0, \dots, m - 1\}\) is greater than \(\max\left(\lambda d-\frac{dm-\frac{\lambda}{2}d(m-q)}{m-i}, \frac{\lambda}{2}d \right)\), provided we cannot find a feasible schedule with \cref{alg:RepairS2-2} (as illustrated in \Cref{fig:k1shelf-schedule}). With this observation, we can analyze the total work area of the schedule and contradict \(\mathcal{W}_1 + \mathcal{W}_2 \leq dm\).

    Consider a schedule as illustrated in \Cref{fig:k1shelf-schedule} (left), where \cref{alg:RepairS2-2} cannot place \(j_0 \in S_2\).
    We start by showing that when placing all jobs in shelf 1 in non-increasing order of processing time, the processing time of jobs (load \(L_i\)) on each machine \(i \in \{0, \dots, m-1\}\) is greater than $g(i) = \lambda d-\frac{dm-\frac{\lambda}{2}d(m-q)}{m-i}$.
    \begin{figure}[h]
        \centering
        \resizebox{\columnwidth}{!}{%
            \begin{tikzpicture}
    \begin{axis}[
            axis lines = middle,
            axis line style = {-},
            xtick=\empty, 
            ytick=\empty, 
            grid = none,
            domain=0:0.9,  
            samples=100,
            ymin=0, 
            ymax=1,
        ]
        \addplot[] {1}
        node[pos=1.3] (S1) {}
        node[pos=0] (1) {};
        \addplot[] {1.46}
        node[pos=1.3] (S2) {}
        node[pos=0] (2) {};
    \end{axis}

    \draw[<->] ($(2)!(0,0)!(2) - (0, 0.2)$) -- ($(S2)!(0,0)!(S2) - (0, 0.2)$)
    node[midway, below] {$m$};
    \draw[-] (1) -- (S1);
    \draw[-, dashed] (S1) -- ++(0.5, 0)
    node[right] {$d$};
    \draw[-] (2) -- (S2);
    \draw[-, dashed] (S2) -- ++(0.5, 0)
    node[right] {$\lambda d$};

    \begin{axis}[
            axis lines = middle,
            axis line style = {-},
            xtick=\empty, 
            ytick=\empty, 
            grid = none,
            domain=0:1,  
            samples=100,
            ymin=0, 
            ymax=1,
            xmin=0,
            xmax=1,
            xscale=1.3
        ]
        \addplot [
            ultra thick,
            orange
        ]
        {max(1.46 - (1 - 0.73 * (1 - 0.3)) / (1 - x), 0.73)}
        node[pos=0] (node0) {}
        node[pos=0.3] (node1) {}
        node[pos=0.6] (node2) {};

    \end{axis}

    \filldraw[fill=red, opacity=0.7] ($(node0) + (0,0.1)$) rectangle ($(node1)!(0,0)!(node1)$);
    \filldraw[fill=red, opacity=0.7] ($(node1) + (0,0.4)$) rectangle ($(node2)!(0,0)!(node2)$);

    \filldraw[fill=yellow, opacity=0.5, dashed] (S2) rectangle ($(node0)$);
    \filldraw[fill=yellow, opacity=0.5, dashed] (S2) rectangle ($(node1)$);

    \draw[<->] ($(node2)!(0,0)!(node2) - (0, 0.5)$) -- ($(S2)!(0,0)!(S2) - (0, 0.5)$)
    node[midway, below] {$q$};

\end{tikzpicture}

\hspace{1cm} 

\begin{tikzpicture}
    \begin{axis}[
            axis lines = middle,
            axis line style = {-},
            xtick=\empty, 
            ytick=\empty, 
            grid = none,
            domain=0:0.9,  
            samples=100,
            ymin=0, 
            ymax=1,
        ]
        \addplot[] {1}
        node[pos=1.3] (S1) {}
        node[pos=0] (1) {};
        \addplot[] {1.46}
        node[pos=1.3] (S2) {}
        node[pos=0] (2) {};
    \end{axis}

    \draw[<->] ($(2)!(0,0)!(2) - (0, 0.2)$) -- ($(S2)!(0,0)!(S2) - (0, 0.2)$)
    node[midway, below] {$m$};
    \draw[-] (1) -- (S1);
    \draw[-, dashed] (S1) -- ++(0.5, 0)
    node[right] {$d$};
    \draw[-] (2) -- (S2);
    \draw[-, dashed] (S2) -- ++(0.5, 0)
    node[right] {$\lambda d$};

    \begin{axis}[
            axis lines = middle,
            axis line style = {-},
            xtick=\empty, 
            ytick=\empty, 
            grid = none,
            domain=0:1,  
            samples=100,
            ymin=0, 
            ymax=1,
            xmin=0,
            xmax=1,
            xscale=1.3
        ]
        \addplot [
            ultra thick,
            orange
        ]
        {max(1.46 - (1 - 0.73 * (1 - 0.3)) / (1 - x), 0.73)}
        node[pos=0] (node0) {}
        node[pos=0.3] (node1) {}
        node[pos=0.6] (node2) {};

    \end{axis}

    \filldraw[fill=red, opacity=0.7] ($(node0) + (0,0.1)$) rectangle ($(node1)!(0,0)!(node1)$);
    \filldraw[fill=red, opacity=0.7] ($(node1) - (0,0.5)$) rectangle ($(node2)!(0,0)!(node2)$);

    \filldraw[fill=yellow, opacity=0.7] (S2) rectangle ($(node1) + (0, 0.1)$);

    \draw[<->] ($(node2)!(0,0)!(node2) - (0, 0.5)$) -- ($(S2)!(0,0)!(S2) - (0, 0.5)$)
    node[midway, below] {$q$};

\end{tikzpicture}%
        }
        \caption{Infeasible (left) and feasible (right) schedule after applying Algorithm \ref{alg:RepairS2-2}. The red line indicates \(\max\left(\frac{\lambda d}{2}, \lambda d-\frac{dm-\frac{\lambda}{2}d(m-q)}{m-i}\right)\). Pink jobs are in \(S_1\) and yellow jobs in \(S_2\).}
        \label{fig:k1shelf-schedule}
    \end{figure}
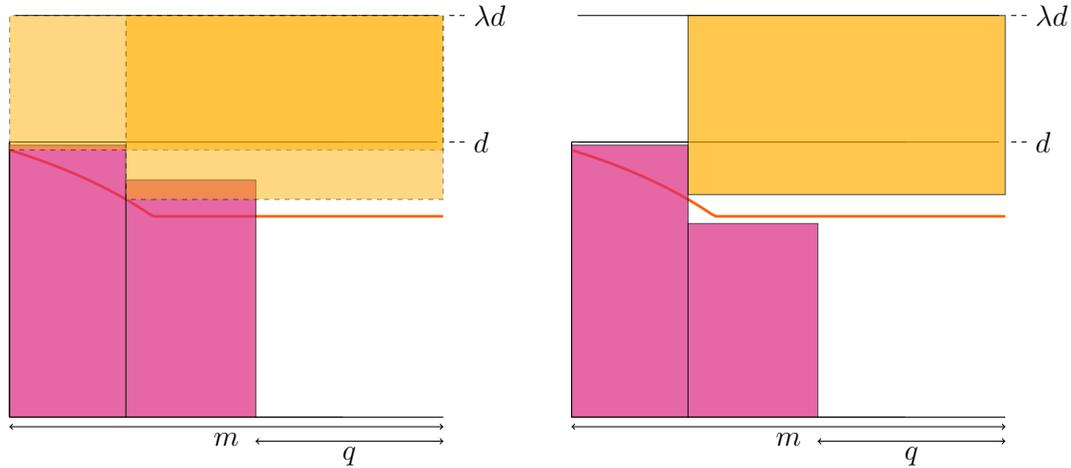

    \begin{restatable}{assumption}{loadAssumption}
        \label{assumption:load}
        If \cref{alg:RepairS2-2} does not find a feasible schedule, the jobs in \(S_1\) can be placed such that the load \(L_i\) on each machine \(i \in \{0, \dots, m - 1\}\) satisfies:
        \begin{align*}
            L_i > g(i) = \lambda d-\frac{dm-\frac{\lambda}{2}d(m-q)}{m-i}
        \end{align*}
    \end{restatable}
    This assumption can be proven by a simple contradiction argument.
    We formally show this in Appendix \ref{sec:app:repair2-2:correctness:loadAssumption}.

    It remains for us to analyze the total work area of the schedule under \cref{assumption:load} to reach a contradiction to \(\mathcal{W}_1 + \mathcal{W}_2 \leq dm\). We will split this analysis into two cases.

    \textit{Case 1 (\(m < 6\))}:
    In this case, we explicitly calculate a bound on \(\mathcal{W}_1 + \mathcal{W}_2\). For more details please refer to Appendix \ref{sec:app:repair2-2:correctness:case2}.

    \textit{Case 2 (\(m \geq 6\))}:
    The most important observation in this case is that we can correctly bound the work on shelf 1 by the following integral:
    \begin{align}
        \label{eq:WBound}
        \mathcal{W}_1 > \int_{0}^{m - q}\max\left\{\lambda d - \frac{dm - \frac{\lambda}{2}d(m -q)}{m - i}, \frac{\lambda}{2}d\right\} \, di
    \end{align}
    This bound is not entirely trivial, since shelf 1 might contain one job with an execution time of less than \(\frac{\lambda}{2}d\), by \cref{lem:repair-properties}.(\ref{lem:repair-properties:1}). We formally show the correctness of this bound, under the assumption of \(m \geq 6\), in Appendix \ref{sec:app:repair2-2:correctness:WBound}.

    In order to reach a contradiction, we need to show that
    \begin{align*}
         & \mathcal{W}_1 + \mathcal{W}_2 \\                                                                                                                                                              & > \int_{0}^{m - q}\max\left\{\lambda d - \frac{dm - \frac{\lambda}{2}d(m -q)}{m - i}, \frac{\lambda}{2}d\right\} \, di + \max\left\{(\lambda - 1)dm, \lambda dq\right\} \\
         & > dm
    \end{align*}
    The first inequation follows directly from \cref{eq:WBound}, \cref{lem:repair-properties}(\ref{lem:repair-properties:2}), and \(\gamma(j_0, (\lambda - 1)d) > m \Rightarrow \mathcal{W}_2 > (\lambda - 1)dm\).
    In order to show the second inequation, we will interpret this expression as a function in \(q\) and show that its minimum, at \(q = \frac{\lambda - 1}{\lambda}m\), is greater than \(dm\).
    \begin{restatable}{observation}{ObservationMinimumWq}
        \label{obs:minimum:Wq}
        The function
        \[
            \mathcal{W}(q) = \int_{0}^{m - q}\max\left\{\lambda d - \frac{dm - \frac{\lambda}{2}d(m -q)}{m - i}, \frac{\lambda}{2}d\right\} \, di + \max\left\{(\lambda - 1)dm, \lambda dq\right\}
        \]
        has a minimum at \(q = \frac{\lambda - 1}{\lambda}m\), for any \(\lambda \in [\frac{10}{7}, \frac{3}{2})\).
    \end{restatable}
    We give a full proof of this observation in Appendix \ref{sec:app:repair2-2:correctness:Wq}.

    Since \(\mathcal{W}(q)\) has the minimal value (in the interval \([0, m]\)) at \(q = \frac{\lambda - 1}{\lambda}m\), we need to show that:
    \begin{align}
        \mathcal{W}\left(\frac{\lambda - 1}{\lambda}m\right)                                                                                                              & > dm              \nonumber                \\
        \Leftrightarrow \int_{0}^{m - \frac{\lambda - 1}{\lambda}m}\max\left\{\lambda d - \frac{\frac{1}{2}dm}{m - i}, \frac{\lambda}{2}d\right\} \, di + (\lambda - 1)dm & > dm              \nonumber                \\
        \Leftrightarrow \int_{0}^{m - \frac{\lambda - 1}{\lambda}m}\max\left\{\lambda d - \frac{\frac{1}{2}dm}{m - i}, \frac{\lambda}{2}d\right\} \, di                   & > (2 - \lambda)dm \label{eq:Wq:inequation}
    \end{align}
    In order to further analyze this expression, we need to take a closer look at the integral term. First, let's find the switching point \(i'\) of the max-function, where \(\lambda d - \frac{\frac{1}{2}dm}{m - i'} = \frac{\lambda}{2}d\).
    \begin{align*}
        \lambda d - \frac{\frac{1}{2}dm}{m - i'}     & = \frac{\lambda}{2}d \\
        \Leftrightarrow \frac{\frac{1}{2}dm}{m - i'} & = \frac{\lambda}{2}d \\
        \Leftrightarrow dm                           & = \lambda d (m -i')  \\
        \Leftrightarrow \frac{\lambda - 1}{\lambda}m & = i'
    \end{align*}
    Then, we can split the integral from \cref{eq:Wq:inequation} at \(i'\) and analyze both parts separately.
    \begin{align*}
         & \int_{0}^{m - \frac{\lambda - 1}{\lambda}m}\max\left\{\lambda d - \frac{\frac{1}{2}dm}{m - i}, \frac{\lambda}{2}d\right\} \, di                                                                                                                    \\
         & = \int_{0}^{\frac{\lambda - 1}{\lambda}m}\left(\lambda d - \frac{\frac{1}{2}dm}{m - i}\right) \, di + \int_{\frac{\lambda - 1}{\lambda}m}^{m - \frac{\lambda - 1}{\lambda}m}\frac{\lambda}{2}d \, di                                               \\
         & = \lambda d\int_{0}^{\frac{\lambda - 1}{\lambda}m}1 \, di - \int_{0}^{\frac{\lambda - 1}{\lambda}m}\frac{\frac{1}{2}dm}{m - i} \, di + \left(\left(m - \frac{\lambda - 1}{\lambda}m\right) - \frac{\lambda - 1}{\lambda}m\right)\frac{\lambda}{2}d \\
         & = \left(\lambda - 1\right)dm - \frac{1}{2}dm \int_{0}^{\frac{\lambda - 1}{\lambda}m}\frac{1}{m - i} \, di + \frac{2 - \lambda}{2}dm                                                                                                                \\
         & = \left(\lambda - 1\right)dm - \frac{1}{2}dm \left(\ln(m)-\ln\left(m - \frac{\lambda - 1}{\lambda}m\right)\right) + \frac{2 - \lambda}{2}dm                                                                                                        \\
         & = \left(\lambda - 1\right)dm - \frac{1}{2}dm \ln(\lambda)+ \frac{2 - \lambda}{2}dm
    \end{align*}

    Together with \cref{eq:Wq:inequation}, we get:
    \begin{align*}
        \left(\lambda - 1\right)dm - \frac{1}{2}dm\ln\left(\lambda\right) + \frac{2 - \lambda}{2}dm & > (2 - \lambda)dm  \\
        \Leftrightarrow - dm\ln\left(\lambda\right) + \left(2 - \lambda\right) dm                   & > (6 - 4\lambda)dm \\
        \Leftrightarrow \ln(\lambda)                                                                & < 3\lambda - 4     \\
        \Leftrightarrow \lambda                                                                     & < e^{3\lambda - 4} \\
        \Leftrightarrow \lambda e^{-3\lambda}                                                       & < e^{-4}           \\
        \Leftrightarrow -3\lambda e^{-3\lambda}                                                     & > -3e^{-4}
    \end{align*}
    The Lambert W function is defined such that \(ye^y=x\), iff \(W(x)=y\). We note that we use the \(W_{-1}\) branch of the Lambert function and since this branch is monotonically decreasing we get, with \(x:=-3e^{-4}\), \(y:=-3\lambda\), the following inequation.
    \begin{align*}
        -3\lambda               & < W_{-1}(-3e^{-4})              \\
        \Leftrightarrow \lambda & > -\frac{1}{3} W_{-1}(-3e^{-4})
    \end{align*}
    Thus, we reach a contradiction to \(\mathcal{W}_1 + \mathcal{W}_2 > dm\) for any \(\lambda > -\frac{1}{3} W_{-1}(-3e^{-4}) \approx 1.4593\).

    This concludes the proof of \cref{lem:RepairS2-2:correctness}.
\end{proof}

\begin{observation}
    \label{obs:contiguous:3}
    The feasible schedule constructed by \cref{alg:RepairS2-2} can be rearranged to a contiguous schedule without altering the makespan.
\end{observation}
\begin{proof}
    Note that by definition, any job contained entirely in one shelf can be scheduled on a contiguous set of machines.
    There exists at most one job \(j_3\), mentioned in \cref{rem:noncontiguous}, that might be split between two shelves.
    A key observation for this proof is that, although the order of jobs in \(S_1\) is important, we can partition the jobs in \(S_1\) into two groups. The first group contains all jobs that fit underneath the singular job in \(S_2\) and the second group contains all jobs that do not fit underneath this job. The order within these two groups does not matter.
    Thus, we can place the job \(j_3\) on the edge of shelf 1, and if required, the allotment can be mirrored.
    With this we can guarantee that the job \(j_3\) is scheduled next to its counterpart in \(S_0\) on adjacent machines.
\end{proof}

\cref{lem:RepairS2-1:correctness,lem:RepairS2-2:correctness}, and \cref{obs:contiguous:3} complete the proof of \cref{thm:dual-approx} and show that we can find a feasible contiguous schedule for all jobs in \(J_B\). The previously ignored small jobs in \(J_S\) can be added greedily to this schedule. We give a more precise description of this process in Appendix \ref{sec:app:small-jobs}.

The time complexity of the presented algorithm is \(O(nm)\).
With a dual approximation framework, we can find a feasible schedule with makespan at most \((1.4593 + \varepsilon)OPT\) in time \(O(nm \log 1/\varepsilon)\).

\section{Experimental Evaluation}
\label{sec:experiments}
In addition to the theoretical improvement on the current state-of-the-art algorithms, we also implemented the presented algorithm.
In the following, we present some experimental results of the algorithm on randomly generated instances and show that the algorithm is capable of solving realistic instances in reasonable time.
The implementation is in Java and publicly available on GitHub (https://anonymous.4open.science/r/MoldableJobScheduling-9AE3) as well as the test-instances and results.
We tested on randomly generated instances with \(n\) jobs and \(m\) machines, where the execution time of each job is uniformly distributed in \([1, 100]\) (on one machine) and the remaining values are uniformly distributed, respecting monotony restrictions.

All tests were run on a single core of an Intel i5-8400T CPU with 2GB of RAM available to the Java process.

We split the tests into two groups, one where we fixed the number of jobs to \(n = 1000\) and varied the number of machines \(m \in [500, 2000]\), and one where we fixed the number of machines to \(m = 1000\) and varied the number of jobs \(n \in [500, 2000]\).
For all tests, we set \(\varepsilon = 0.05\).

\Cref{fig:runtime_machines,fig:runtime_jobs} show the execution time of our implementation for these two groups of tests.
Although we did not take any special care to optimize the implementation, the algorithm is able to solve all instance in less than 30 seconds. The execution time decreases for \(n >> m\), since in these instances the optimal makespan is bigger. This results in more jobs being classified as small jobs and an easier to solve instance.

\begin{figure}[H]
    \centering
    \includegraphics[width=\textwidth]{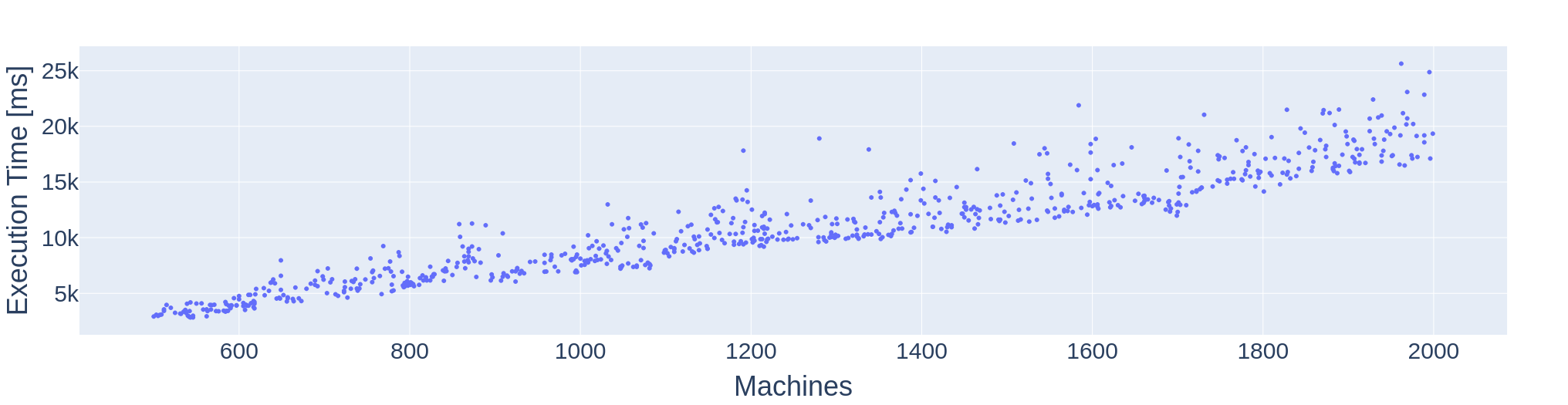}
    \caption{Runtime of the algorithm for varying number of machines.}
    \label{fig:runtime_machines}
\end{figure}
\begin{figure}[H]
    \centering
    \includegraphics[width=\textwidth]{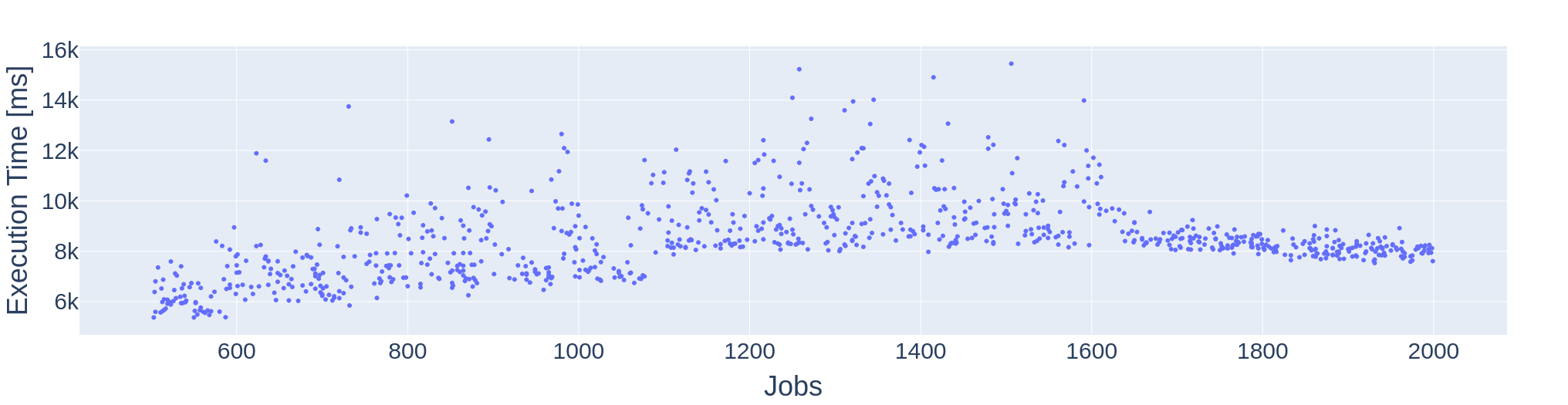}
    \caption{Runtime of the algorithm for varying number of jobs.}
    \label{fig:runtime_jobs}
\end{figure}

We note that while the theoretical worst-case approximation ratio is \(1.4593 + \varepsilon\), the algorithm was able to guarantee a makespan below \((\frac{10}{7} + \varepsilon)OPT\) in all randomly generated instances.
This is due to the fact that the worst-case approximation ratio only occurs in very specific cases.
To be more precise, the presented algorithm can guarantee a makespan of at most \((\frac{13}{9} + \varepsilon)OPT\), unless the number of idle machines in shelf 1 is greater than \(m/6\) (see \cref{lem:RepairS2-2:correctness}).

However, we can show that our analysis is tight in the worst case.
Consider the following instance with \(m = 13\) and \(n=10\).
All jobs have a constant work-function with the following values:
\begin{itemize}
    \item \(j_1\): \(w(j, k) = 6.01, \quad\forall k \in \{1, \dots, m\}\)
    \item \(j_2\): \(w(j, k) = 0.99, \quad\forall k \in \{1, \dots, m\}\)
    \item \(j_3-j_{10}\): \(w(j, k) = \frac{3}{4}, \quad\forall k \in \{1, \dots, m\}\)
\end{itemize}

An optimal schedule for this instance has a makespan of \(C_{\max} = 1\), since the work of all jobs is \(W = 13\) and by placing all jobs on \(m = 13\) machines, we can achieve a makespan of \(OPT = W/m = 1\).
On the other hand, an optimal solution to the \ac{MCKP} problem in our algorithm is given by \(\mathcal{C}_1 = \{j_2, \dots, j_{10}\}\), \(\mathcal{C}_2 = \emptyset\), \(\mathcal{C}_3 = \{j_1\}\).
In this case, the algorithm will schedule \(j_1\) on \(13\) machines and all remaining jobs on \(1\) machine each.
This results in a makespan of \(C_{\max} = 6.01/13 + 0.99 \approx 1.452\).
For larger \(m\), we can construct similar instances and achieve an approximation ratio arbitrarily close to our theoretical worst-case guarantee of \(1.4593 + \varepsilon\).
Due to the very specific structure of these instances, they are highly unlikely to occur in practice.

%
%

\section{Conclusion}
The algorithm presented in this work achieves an approximation ratio slightly below \((73/50 + \varepsilon)\) for any \(\varepsilon > 0\) in time complexity \(O(nm \log 1/\varepsilon)\) for the problem of scheduling monotone moldable jobs. Even more surprisingly, we also show how to apply this algorithm to the contiguous variant of the problem.
Since we construct the same schedule for the contiguous variant as for the non-contiguous variant, this bounds the gap between these two variants, which is an interesting side effect of our algorithm.

We show in an implementation that the algorithm achieves a ratio of \(10/7\) in most cases and the worst-case approximation ratio of \(1.4593 + \varepsilon\) only happens in very specific problematic cases.

\subsection{Future Work}
The presented algorithm breaks the long-standing barrier of a practically efficient algorithm with an approximation ratio below \(1.5\) for the problem of scheduling monotone moldable jobs.

We hope that this work will lead to further improvements in the field of scheduling monotone moldable jobs.
In particular, we leave the following open questions for future work:

\begin{itemize}
    \item \emph{Can the approximation ratio be improved further?}
          There is no known lower bound on the approximation ratio for practically efficient algorithms. As we showed, the presented algorithm already achieves an approximation ratio below the worst case guarantee in most cases.
          This question is specifically interesting for the contiguous variant, since no \ac{PTAS} is known.
    \item \emph{Can the running time be improved?}
          Jansen and Land \cite{jansenland2018schedulingmoldable} introduced the concept of \emph{compression} to achieve a logarithmic dependence on the number of machines \(m\) in the running time of the algorithm.
          This was improved by Grage, Jansen, and Ohnesorge \cite{grage2023improved} by utilizing (min, +)-convolution.
          These techniques could potentially be applied to our algorithm to reduce the dependence on \(m\) in the running time.
    \item \emph{Can this result be applied to related problems?}
          The presented algorithm could lead to improvements in other the closely related problems such as CPU/GPU scheduling~\cite{germann20243}, 2D Knapsack~\cite{DBLP:conf/focs/GalvezGHI0W17}, strip packing~\cite{DBLP:conf/wads/HarrenJPS11}, or demand strip packing~\cite{DBLP:conf/soda/EberleHRW25} (for practically efficient algorithms).
\end{itemize}
\bibliography{src}

\newpage
\appendix

\section{Proof of \cref{lem:repair-properties}}
\label{sec:app:repair2:properties}

We start by applying transformations (\cref{alg:repair1}) to the 3-shelf schedule, which will help reduce the number of machines used by \(S_2\). These transformations are based on ideas from Mounié, Rapine, and Trystram \cite{mounie2007} but are parameterized with \(\lambda\).

\renewcommand{\labelenumi}{T\arabic{enumi}}
\renewcommand{\theenumi}{T\arabic{enumi}}
\begin{enumerate}
    \item \label{trans:1}If a job \(j \in S_1\) has execution time at most \(\frac{\lambda}{2}d\) and is allotted to \(p > 1\) machines, allocate \(j\) to \(\gamma(j, \lambda d)\) machines in \(S_0\).
    \item \label{trans:2}If \(j, j' \in S_1\) have execution time less than \(\frac{\lambda}{2}d\) and are each allotted to \(1\) machine, allocate \(j\) and \(j'\) to the same machine in \(S_0\).
    \item \label{trans:3}Let \(q\) denote the number of idle machines in \(S_1\). If there exists a job \(j \in S_2\) with an execution time of less than \(\lambda d\) on \(q\) machines, allocate \(j\) on \(\gamma(j, \lambda d)\) machines either to \(S_1\) or \(S_0\), depending on its execution time.
\end{enumerate}

\begin{algorithm}[H]
    \caption{Repair Shelf Schedule}\label{alg:repair1}
    \begin{algorithmic}[1]
        \Require 3-shelf schedule; \(\lambda \in [\frac{10}{7}, \frac{3}{2})\).
        \While{possible}
        \State apply Transformation \ref{trans:1}, \ref{trans:2}, or \ref{trans:3}.
        \EndWhile
    \end{algorithmic}
\end{algorithm}

In the following, we will analyze the schedule that results after applying \Cref{alg:repair1}.
Remember, that for ease of notation, we assume \(m_0 = 0\).
This does not change the correctness of the proof, since we contradict the total work of the schedule, and shelf 0 has a total work greater than \(dm_0\).

\LemmaRepairProperties*
\begin{proof}
    We prove the lemma by first showing that Properties (\ref{lem:repair-properties:1}) and (\ref{lem:repair-properties:2}) hold after applying \cref{alg:repair1}.
    Then, we will establish the contiguity of the schedule and analyze the time complexity of the algorithm.

    We begin by showing the three properties of the lemma.
    By the definition of \cref{alg:repair1}, any job \(j \in S_1\) is scheduled on \(\gamma(j, d)\) machines.
    This implies that any job with an execution time of less than \(\frac{\lambda}{2}d\) is scheduled on exactly one machine, by the definition of \(\gamma(j, d)\) and monotony.
    The schedule contains at most one such job, since we could otherwise apply Transformation \ref{trans:2}, contradicting the algorithms' termination condition.
    Property (\ref{lem:repair-properties:2}) follows directly, since Transformation \ref{trans:3} cannot be applied.

    For the last property we use a proof by contradiction: Suppose Property (\ref{lem:repair-properties:3}) does not hold. This implies that there exists a job \(j \in S_1\) with an execution time of \(t \in (\frac{3}{7}d, \frac{\lambda}{2}d]\).
    By Property (\ref{lem:repair-properties:1}), there exists at most one such job, which is scheduled on exactly one machine.
    We consider three cases:

    \textit{Case 1 (\(S_1 \setminus \{j\} = \emptyset\) and \(q = 0\))}: In this case, \(m = 1\). We place all jobs on this one machine and obtain a feasible schedule, if \(\mathcal{W}_1 + \mathcal{W}_2 \leq d\).

    \textit{Case 2 (\(S_1 \setminus \{j\} = \emptyset\) and \(q > 0\))}: In this case, \(q = m - 1\).
    Since $S_2$ contains at least one job $\mathcal{W}_2 > \lambda dq = \lambda d(m - 1)$, by Property (\ref{lem:repair-properties:2}). With $\mathcal{W}_1 + \mathcal{W}_2 \leq dm$, we get:
    \begin{align*}
        \frac{3}{7}d + \lambda d(m - 1)                   & < dm                                        \\
        \Leftrightarrow \frac{3}{7} + \lambda m - \lambda & < m                                         \\
        \Leftrightarrow (\lambda - 1)m                    & < \lambda - \frac{3}{7}                     \\
        \Leftrightarrow m                                 & < \frac{\lambda - \frac{3}{7}}{\lambda - 1}
    \end{align*}
    For \(\lambda \in [\frac{10}{7}, \frac{3}{2})\), this implies \(m \leq 2\).
    Since \(q \geq 1\) and \(m = q + 1\), we must have \(m=2\) and \(q=1\).
    We will first show that $S_2$ can contain at most one job. To this end, denote the number of jobs in $S_2$ as $k$, then the total work area of $S_2$ is $\mathcal{W}_2 > \lambda dqk = \lambda dk$, by Property~(\ref{lem:repair-properties:2}) and \(q = 1\). With $\mathcal{W}_1 + \mathcal{W}_2 \leq dm$ we get:
    \begin{align*}
        \frac{3}{7}d + \lambda dk & < 2d                                                                                  \\
        \Leftrightarrow k         & < \frac{11}{7\lambda} < \frac{11}{10} < 2 &  & \text{(\(\lambda \geq \frac{10}{7}\))}
    \end{align*}
    This leads to a final contradiction.
    If we cannot find a feasible schedule, the single job in \(S_2\) has a processing time of at least \(\lambda d - t\) on two machines.
    This implies a total work area of \(2(\lambda d - t) + t \geq \frac{3}{2}\lambda d > 2d\), where the first inequality follows from \(t \leq \frac{\lambda}{2}d\) and the second from \(\lambda > \frac{10}{7}\).
    Since \(\mathcal{W}_1 + \mathcal{W}_2 \leq 2d\), this is a contradiction.

    \textit{Case 3 (\(S_1 \setminus \{j\} \neq \emptyset\))}: If Property (\ref{lem:repair-properties:3}) does not hold, we may assume that there exists another job \(j' \in S_1 \setminus \{j\}\) with an execution time less than \(\lambda d - t\).
    Then, we can schedule \(j\) on top of \(j'\).
    Observe that \(t(j, d) + t(j', d) > d\), therefore we place \(j\) (and \(j'\) partially) in \(S_0\). The new schedule satisfies (\ref{lem:repair-properties:3}).

    It remains to show that the resulting schedule is contiguous.
    We must ensure that no job is split across non-adjacent machines.
    Note that this can only happen for jobs that are split between two shelves.
    If such a job \(j_3\) exists (see \cref{rem:noncontiguous}), we can choose \(j\) (the job violating Property (\ref{lem:repair-properties:3})) to be \(j \coloneqq j_3\).
    We consider the three cases from above:
    In Cases 1 and 2, job \(j\) (\(=j_3\)) is the only job in \(S_1\).
    We can schedule it at the boundary between shelf 1 and shelf 0, ensuring it occupies a contiguous block of machines.
    In case 3, \(j\) (\(=j\)) is moved entirely into shelf 0, so it is no longer split.
    The other job \(j'\) can be chosen as the smallest job in \(S_1\) and placed at the boundary, preserving contiguity.
    This guarantees that all jobs are scheduled on adjacent machines.

    The time complexity results from the following argument:
    Denote the number of big jobs as \(n_B\).
    Since the transformations move jobs from shelf 2 to shelf 1 or 0, and from shelf 1 to shelf 0, the while-loop in \cref{alg:repair1} can be executed at most \(O(n_B)\) times.
    Transformations~\ref{trans:1}~and~\ref{trans:2} can be implemented in \(O(n_B)\). Transformation \ref{trans:3} can be implemented in \(O(n_B)\) for scanning the jobs in \(S_2\) and \(O(\log m)\) to determine \(\gamma(j, \lambda d)\) for the chosen job \(j \in S_2\). This results in a time complexity of \(O(n_B^2 + n_B \log m)\). Since \(n_B\) is bounded by both \(n\) and \(m\), we can conclude that the time complexity of \cref{alg:repair1} is \(O(nm)\).
\end{proof}

\section{Proof of \cref{lem:RepairS2-1:correctness}}
\label{sec:app:repair2-1:correctness}

\LemmaRepairCorrectness*

\begin{algorithm}[H]
    \caption{RepairS2-1}\label{alg:RepairS2}
    \begin{algorithmic}[1]
        \While {$m_2 > m$}
        \State Let $j$ be the job with the smallest height among jobs in $S_2$.
        \State Let $p$ be the number of machines $j$ is allotted to.
        \State Allot $j$ to $p - 1$ machines.
        \EndWhile
        \State Sort jobs in $S_1$ in descending order of execution time.
        \State Sort jobs in $S_2$ in ascending order of execution time.
    \end{algorithmic}
\end{algorithm}
The feasible schedule is constructed, by applying \cref{alg:RepairS2}. During the procedure, we iteratively compress the job with the smallest height among jobs in \(S_2\) by one machine until the number of machines in \(S_2\) is at most \(m\). We will show that the jobs of \(S_1\) and \(S_2\) do not intersect, when sorted in descending (ascending) order of execution time in shelf 1 (shelf 2).
The termination of \cref{alg:RepairS2} is proven in Appendix \ref{sec:app:repair2:termination}.
We separate the proof of \cref{lem:RepairS2-1:correctness} into two cases, depending on the value of \(q\).

\textit{Case 1 (\(q = 0\)):}
This is the simplest case in which we will show that we can obtain a feasible schedule with makespan at most \(\lambda d\), for any \(\lambda \geq \frac{10}{7}\), if \(\mathcal{W}_1 + \mathcal{W}_2 \leq dm\).

\begin{observation}
    \label{obs:tS2:upper:afterAlg}
    After applying \cref{alg:RepairS2}, the height of any job $j \in S_2$ is at most $\frac{4}{7}d$.
\end{observation}
\begin{proof}
    Suppose there exists a job \(j\) with execution time strictly greater than \(\frac{4}{7}d\) in \(S_2\). Due to work monotony this job had a height of at least \(\frac{2}{7}d\), before its last compression during \cref{alg:RepairS2}.
    Since the algorithm selects the smallest job for compression, all other jobs in \(S_2\) have a height of at least \(\frac{2}{7}d\), thus \(\mathcal{W}_2 > \frac{2}{7}dm_2\). With \cref{lem:repair-properties}.(\ref{lem:repair-properties:3}) we get: \(dm \geq \mathcal{W}_2 + \mathcal{W}_1 > \frac{2}{7}dm_2 + \frac{5}{7}dm > dm\). A contradiction.
\end{proof}

\begin{lemma}
    \label{lem:RepairS2-1-1:correctness}
    If \(q = 0\), \cref{alg:RepairS2} delivers a feasible schedule with makespan at most \(\frac{10}{7}d\), unless $\mathcal{W}_1 + \mathcal{W}_2 > dm$.
\end{lemma}
\begin{proof}
    For the sake of contradiction, let's assume that Algorithm \ref{alg:RepairS2} delivers an infeasible schedule. This implies that there exists a machine $i$ with load $L_i > \frac{10}{7}d$, where we denote the height of $S_2$ on that machine as $t$.

    Let us first establish that $\frac{3}{7}d \leq t \leq \frac{4}{7}d$. Indeed, the lower bound is straightforward, since the load on this machine is greater than $\frac{10}{7}d$ and any job on $S_1$ has an execution time at most $d$. The upper bound is stated in Observation \ref{obs:tS2:upper:afterAlg}.

    We can now proceed to analyze the work of this schedule. The height of each job in $S_2$ to the right of $i$ is at least $t$ since they are placed in ascending order of processing time and each job (apart from at most one) in $S_1$ has height at least $\frac{\lambda}{2}d = \frac{5}{7}d$, by \cref{lem:repair-properties}.(\ref{lem:repair-properties:1}). With this, we get:
    \begin{align*}
        L_{i'} \geq t + \frac{5}{7}d \geq \frac{8}{7}d \quad \forall i' \in \{i+1, \dots, m - 1 \}
    \end{align*}
    The height of each job in $S_1$ to the left of $i$ is at least $\frac{10}{7}d - t$ since they are placed in descending order of processing time and each job in $S_2$ has an execution time greater than $\frac{3}{14}d$, since they are placed on \(\gamma(j, \frac{3}{7}d)\) machines. With this, we get:
    \begin{align*}
        L_{i'} \geq \frac{10}{7}d - t + \frac{3}{14}d \geq \frac{15}{14}d \quad \forall i' \in \{1, \dots, i - 1 \}
    \end{align*}
    Any of these machines has a load strictly greater than \(d\). It remains to consider the two left-out machines, $i$ and \(m\).

    The load on machine \(m\) is at least \(t + \frac{3}{7}d \geq \frac{6}{7}d\), since any job in \(J_B\) has a sequential height of at least \(\frac{3}{7}d\) on one machine. The load on machine \(i\) is at least \(\frac{10}{7}d\). Thus, their average load is also strictly greater than \(d\), which contradicts \(\mathcal{W}_1 + \mathcal{W}_2 \leq dm\).
\end{proof}

\begin{observation}
    \label{obs:contiguous:1}
    If \(q = 0\), the feasible schedule constructed by \cref{alg:RepairS2} can be modified to be contiguous.
\end{observation}
\begin{proof}
    We recall that there exists at most one job that is split between two shelves. This job \(j_3\), mentioned in \cref{rem:noncontiguous}, is split between \(S_0\) and \(S_1\) with a height of at least \(\frac{5}{7}d\) and at most \(\frac{6}{7}d\). Since the height of any job in \(S_2\) is at most \(\frac{4}{7}d\) (by \cref{obs:tS2:upper:afterAlg}), the order of jobs in shelf 1 with height at most \(\frac{4}{7}d\) does not matter. Thus, we can move \(j_3\) to the rightmost machine in shelf 1 and schedule it on the edge of shelf 1. The remaining part of \(j_3\) will be scheduled on the leftmost machine in shelf 1. This guarantees that all jobs are scheduled on adjacent machines.
\end{proof}

\textit{Case 2 (\(0 < q \leq m/6\)):}
Again, we prepare the proof of this case by stating a few observations that will be used throughout the proof.
\begin{observation}
    \label{obs:W_2:upper}
    Given \(\lambda \geq \frac{13}{9}\) and \(q \leq m/6\), the total work area of $S_1$ is greater than $\frac{65}{108}dm$. And the total work area of $S_2$ is less than $\frac{43}{108}dm$.
\end{observation}
\begin{proof}
    With \cref{lem:repair-properties}.(\ref{lem:repair-properties:3}) and \(\lambda \geq \frac{13}{9}\), we get $\mathcal{W}_1 > \frac{13}{18}d(m - q)$. With $q \leq \frac{1}{6}m$, we get:
    \begin{align*}
        \mathcal{W}_1 > \frac{13}{18}d(m - q) \geq \frac{13}{18}dm - \frac{13}{108}dm = \frac{65}{108}dm
    \end{align*}
    The upper bound for $\mathcal{W}_2$ follows directly from $\mathcal{W}_1 + \mathcal{W}_2 \leq dm$.
\end{proof}

\begin{observation}
    \label{obs:min3machines}
    For any \(\lambda \geq \frac{13}{9}\), after applying Algorithm \ref{alg:RepairS2}, any job $j \in S_2$ is allotted to at least $3$ machines.
\end{observation}
\begin{proof}
    We will first argue that before the application of \cref{alg:RepairS2}, any job $j \in S_2$ is allotted to at least $3$ machines and then proceed to show that this property is maintained during the application of the algorithm.
    By \cref{lem:repair-properties}.(\ref{lem:repair-properties:2}) and \(q \geq 1\), each job in \(j \in S_2\) has a work area greater than \(\frac{13}{9}d\). Thus, by monotony and the definition of \(\gamma\), \(\gamma(j, \frac{4}{9}d) > 3\).

    Now suppose, at any point during the algorithm, a job \(j \in S_2\) is compressed from three to two machines. Using work-monotony, we know that \(t(j, 3) > \frac{13}{9}d / 3 = \frac{13}{27}d\). Since the algorithm selects the smallest job, all other jobs must be at least as large. This implies that the work area of \(S_2\) is greater than \(\frac{13}{27}dm > \frac{43}{108}dm\), contradicting \cref{obs:W_2:upper}.
\end{proof}

\begin{observation}
    \label{obs:S2t:lower:afterAlg}
    All jobs in $S_2$ have an execution time greater than $\frac{12}{35}d$.
\end{observation}
\begin{proof}
    Consider any job \(j \in S_2\), by $q \geq 1$ and \cref{lem:repair-properties}.(\ref{lem:repair-properties:2}), the work area of this job is at least $\frac{13}{9}d$. We split this proof into $2$ cases based on the number of assigned machines. Suppose $\gamma(j, \frac{3}{7}d) \leq 4$. Then its height is at least $\frac{13}{9}d / 4 = \frac{13}{36}d > \frac{12}{35}d$.
    If $\gamma(j, \frac{3}{7}d) > 4$, then by the definition of $\gamma$ and monotony:
    \[t\left(j, \gamma\left(j, \frac{3}{7}\right)\right) > \frac{3}{7}d \cdot \frac{\gamma(j, \frac{3}{7}d) - 1}{\gamma(j, \frac{3}{7}d)} \geq \frac{3}{7}d \cdot \frac{4}{5} = \frac{12}{35}d\]
\end{proof}

\begin{observation}
    \label{obs:tS2:upper:afterAlg-2}
    After applying Algorithm \ref{alg:RepairS2}, the height of any job $j \in S_2$ is at most $\frac{43}{81}d$.
\end{observation}
\begin{proof}
    Suppose $S_2$ contains a job $j$ with a height greater than $\frac{43}{81}d$. Let \(\alpha(j)\) denote the number of machines assigned to \(j\).
    By Observation \ref{obs:min3machines} we know that \(\alpha(j) \geq 3\). Using the same argument as in the proof of \cref{obs:S2t:lower:afterAlg}, the height of this job on one additional machine is \(t(j, \alpha(j) + 1) \geq \frac{43}{81}d \cdot \frac{\alpha(j)}{\alpha(j) + 1} \geq \frac{43}{81}d \cdot \frac{3}{4} = \frac{43}{108}d\). Since any job in \(S_2\) must be at least this tall, this is a contradiction to \cref{obs:W_2:upper}.
\end{proof}

\begin{lemma}
    \label{lem:RepairS2-1-2:correctness}
    If \(0 < q \leq m/6\), \cref{alg:RepairS2} delivers a feasible schedule with makespan at most \(\frac{13}{9}d\), unless $\mathcal{W}_1 + \mathcal{W}_2 > dm$.
\end{lemma}
\begin{proof}
    For the sake of contradiction, assume that Algorithm \ref{alg:RepairS2} delivers an infeasible schedule. This implies that there exists a machine $i$ with load $L_i > \frac{13}{9}d$, where we denote the height of $S_2$ on that machine as $t$. Note that $t \in (\frac{4}{9}d, \frac{43}{81}d)$.

    We can now proceed to analyze the work of this schedule. The height of each job in $S_1$ to the left of $i$ is greater than $\frac{13}{9}d - t$ since they are placed in descending order of processing time and each job in $S_2$ has a height of at least $\frac{12}{35}d$, by Observation \ref{obs:S2t:lower:afterAlg}. With this, we get:
    \begin{align}
        L_{i'} \geq \frac{13}{9}d - t + \frac{12}{35}d > \frac{73}{63}d \quad \forall i' \in \{1, \dots, i - 1\}
    \end{align}
    The height of any job in $S_2$ on the last $m - i$ machines is at least \(t > \frac{4}{9}d\) since they are placed in ascending order of processing time. With each job on $S_1$ (apart from at most one) having an execution time at least $\frac{13}{18}d$, by \cref{lem:repair-properties}.(\ref{lem:repair-properties:1}), we get the following lower bound for the load of all machines except the last \(q + 1\):
    \begin{align}
        L_{i'} \geq \frac{4}{9}d + \frac{13}{18}d = \frac{73}{63}d \quad \forall i' \in \{i + 1, \dots, m - q - 1\}
    \end{align}
    Machine \(i\) has, by selection, a load strictly greater than \(\frac{13}{9}d\) and the least loaded machine \(q\) in \(S_1\) might contain a job with a height only greater than \(\frac{3}{7}d\). Thus, the average load of these two machines is greater than: \((\frac{13}{9}d + \frac{4}{9}d + \frac{3}{7}d)/2 = \frac{73}{63}d\).
    With the last $q$ machines having a load strictly greater than $\frac{4}{9}d$, we get the following inequality:
    \begin{align*}
        \frac{73}{63}d \cdot (m - q) + \frac{4}{9}dq & < dm           \\
        \Leftrightarrow q                            & > \frac{2}{9}m
    \end{align*}
    This is a contradiction to $q \leq \frac{m}{6}$.
\end{proof}

\begin{observation}
    \label{obs:contiguous:2}
    If \(0 < q \leq m/6\), the feasible schedule constructed by \cref{alg:RepairS2} can be modified to be contiguous.
\end{observation}
\begin{proof}
    We recall that there exists at most one job that is split between two shelves. This job \(j_3\), mentioned in \cref{rem:noncontiguous}, is split between \(S_0\) and \(S_1\) with a height of at least \(\frac{5}{7}d\) and at most \(\frac{6}{7}d\). Since the height of any job in \(S_2\) is at most \(\frac{43}{81}d\) (by \cref{obs:tS2:upper:afterAlg-2}), the order of jobs in shelf 1 with height at most \(\frac{13}{9}d - \frac{43}{81}d = \frac{74}{81}d\) does not matter. Thus, we can move \(j_3\) to the rightmost machine in shelf 1 and schedule it on the edge of shelf 1. The remaining part of \(j_3\) will be scheduled on the leftmost machine in shelf 1. This guarantees that all jobs are scheduled on adjacent machines.
\end{proof}

\cref{lem:RepairS2-1-1:correctness,lem:RepairS2-1-2:correctness}, and Observations \ref{obs:contiguous:1} and \ref{obs:contiguous:2} complete the proof of \cref{lem:RepairS2-1:correctness}.

\section{Proof of Termination of \cref{alg:RepairS2}}
\label{sec:app:repair2:termination}
In order to prove that \cref{alg:RepairS2} terminates, we will show that the number of machines in \(S_2\) is bounded by \(2m\). We first note that any job in \(J_B\) has a sequential execution time strictly greater than \(\frac{3}{7}d\), by definition.
Further, there must exist a job that is assigned to at least two machines in \(S_2\) as long as \(m_2 > m\), since otherwise we would directly get \(\mathcal{W}_2 > \frac{3}{7}dm\).
Thus, we just need to show the following observation:
\begin{observation}
    \label{obs:machinesS2:bound}
    Given a 3-shelf schedule with \(q \leq \frac{m}{6}\), the number of machines \(m_2\) used by shelf 2 is less than \(3m\).
\end{observation}
\begin{proof}
    Suppose \(m_2 \geq 3m\). Observe that any job on shelf 2 has a height of at least \(\frac{\lambda - 1}{2}d\), since it is placed on \(\gamma(j, (\lambda - 1)d)\) machines. With \(\mathcal{W}_1 + \mathcal{W}_2 \leq dm\) and \cref{lem:repair-properties}.(\ref{lem:repair-properties:3}), we get:
    \begin{align*}
        \mathcal{W}_1 + \mathcal{W}_2
         & \geq \frac{\lambda}{2}d(m - q) + \frac{(\lambda - 1)}{2}dm_2                                                      \\
         & \geq \frac{5}{12}\lambda dm + \frac{3(\lambda - 1)}{2}dm     &  & \text{(\(q \leq \frac{m}{6}\);\(m_2 \geq 3m\))} \\
         & > dm                                                         &  & \text{(\(\lambda \geq \frac{10}{7}\))}
    \end{align*}
    A contradiction.
\end{proof}
With this observation, we can conclude that \cref{alg:RepairS2} terminates within \(O(mn)\) iterations.

\section{Omitted Proofs of \cref{lem:RepairS2-2:correctness}}
\label{sec:app:repair2-2:correctness:case3}

To prepare the proof of the following statements, we start by stating an observation.
Notice that this observation is basically a reformulation of \cref{lem:repair-properties}.
Although, we will mention transformations in the proof of this observation, these transformations have already been applied in \cref{alg:repair1}.
Therefore, we do not actually have to alter the schedule in order for this observation to hold.
\begin{observation}
    \label{lem:avgL:2m:S1}
    The average load of any two non-idle machines in $S_1$ is greater than $\frac{\lambda}{2}d$.
\end{observation}
\begin{proof}
    We begin this proof by stating that any job with an execution time smaller or equal to $\frac{\lambda}{2}d$ on $S_1$ has to be assigned to one machine, by the definition of $\gamma$.
    We will also note that there exists at most one such job, otherwise, we could apply Transformation \ref{trans:2}.

    Now suppose there exists a job $j$ with an execution time smaller than $\frac{\lambda}{2}d$ and another job $j'$ such that $t(j, \gamma(j, d)) + t(j', \gamma(j', d)) \leq \lambda d$. Then we could place $j$ on top of $j'$ and place them in $S_0$, while $j'$ remains partially in $S_1$.
\end{proof}

\subsection{Proof of \cref{assumption:load}}
\label{sec:app:repair2-2:correctness:loadAssumption}
\loadAssumption*
\begin{proof}
    Suppose there exists a machine $i$ with for which this is not the case.
    We will now show that the singular job in $S_2$ can be placed on the $m - i$ least loaded machines without exceeding the total makespan of $\lambda d$, as illustrated in \Cref{fig:k1shelf-schedule} (right).

    Using work monotony, \(\mathcal{W}_2 \leq dm - \mathcal{W}_1\), and \cref{lem:repair-properties}.(\ref{lem:repair-properties:3}), we get:
    \begin{align}
        \label{eq:load:bound}
        t(j_0, m - i) \leq \frac{w(j_0, m_2)}{m -i} < \frac{dm-\frac{\lambda}{2}d(m - q)}{m - i}
    \end{align}

    When placing \(j_0\) on the \(m - i\) least loaded machines, no machine has a load greater than \(g(i) + t(j_0, m - i)\), since machine $i$ has the highest load among the \(m - i\) least loaded machines.
    With the definition of \(g\) and \cref{eq:load:bound}, we get:
    \begin{align*}
        g(i) + t(j_0, m - i) < \lambda d
    \end{align*}

    Thus, \(j_0\) would have been scheduled on \(m - i\) machines by \cref{alg:RepairS2-2}, and for the remainder of this proof, we may assume that each machine \(i \in \{0, \dots, m - 1\}\) has jobs of height at least \(g(i) = \lambda d-\frac{dm-\frac{\lambda}{2}d(m-q)}{m-i}\) scheduled on it (these jobs are in \(S_1\)).
\end{proof}

\subsection{Proof of \cref{eq:WBound}}
\label{sec:app:repair2-2:correctness:WBound}
\begin{observation}
    The following expression is a correct bound for the work on shelf 1.
    \begin{align*}
        \mathcal{W}_1 > \int_{0}^{m - q}\max\left\{\lambda d - \frac{dm - \frac{\lambda}{2}d(m -q)}{m - i}, \frac{\lambda}{2}d\right\} \, di
    \end{align*}
\end{observation}
\begin{proof}
    We note that there might exist one job that is smaller than \(\frac{\lambda}{2}d\) on shelf 1, by \cref{lem:repair-properties}.(\ref{lem:repair-properties:1}). However, by \cref{lem:avgL:2m:S1}, this inequality is still true when the load of at least two machines is bounded by the second term in the max-function, since their average load is greater than \(\frac{\lambda}{2}d\). We will now show that this is the case for \(m \geq 6\).

    Let \(x\) be the intersection point of \(\lambda d - \frac{dm - \frac{\lambda}{2}d(m -q)}{m - i}\) and \(\frac{\lambda d}{2}\):
    \begin{align*}
        \lambda d - \frac{dm - \frac{\lambda}{2}d(m -q)}{m - x}      & = \frac{\lambda}{2}d        \\
        \Leftrightarrow \frac{dm - \frac{\lambda}{2}d(m -q)}{m - x}  & = \frac{\lambda}{2}d        \\
        \Leftrightarrow dm - \frac{\lambda}{2}d(m -q)                & = \frac{\lambda}{2}d(m - x) \\
        \Leftrightarrow \frac{2dm - \lambda d(m -q)}{\lambda d}      & = m - x                     \\
        \Leftrightarrow \frac{-2dm + \lambda d(m -q)}{\lambda d} + m & = x                         \\
        \Leftrightarrow -\frac{2m}{\lambda} + 2m - q                 & = x
    \end{align*}

    Notice that the number of machines bounded by the second term in the max-function is \(m - q - x\). We now want to show that \(m - q - x \geq 2\):
    \begin{align*}
        m - q - x                                            & \geq 2                                   \\
        \Leftrightarrow m - q + \frac{2m}{\lambda} - 2m + q  & \geq 2                                   \\
        \Leftrightarrow -m + \frac{2m}{\lambda}              & \geq 2                                   \\
        \Leftrightarrow m\left(-1 + \frac{2}{\lambda}\right) & \geq 2                                   \\
        \Leftrightarrow m \frac{-\lambda + 2}{\lambda}       & \geq 2                                   \\
        \Leftrightarrow m\left(2 - \lambda\right)            & \geq 2\lambda                            \\
        \Leftrightarrow 12 - 6\lambda                        & \geq 2\lambda &  & \text{(\(m \geq 6\))} \\
        \Leftrightarrow \frac{3}{2} \geq \lambda
    \end{align*}
    Thus, the observation holds for \(m \geq 6\), if \(\lambda < \frac{3}{2}\).
\end{proof}

\subsection{Proof of \cref{obs:minimum:Wq}}
\label{sec:app:repair2-2:correctness:Wq}
\ObservationMinimumWq*
To this end, we will show the following observations:
\begin{observation}
    \label{obs:W1:q:1}
    The function
    \begin{align}
        \label{eq:W1:q:1}
        a(q) = \int_{0}^{m - q}\max\left\{\lambda d - \frac{dm - \frac{\lambda}{2}d(m -q)}{m - i}, \frac{\lambda}{2}d\right\} \, di + (\lambda - 1)dm
    \end{align}
    is monotonically decreasing in \(q\) in the interval \([0, \frac{\lambda - 1}{\lambda}m]\).
\end{observation}

\begin{observation}
    \label{obs:W1:q:2}
    The function
    \begin{align}
        \label{eq:W1:q:2}
        b(q) = \int_{0}^{m - q}\max\left\{\lambda d - \frac{dm - \frac{\lambda}{2}d(m -q)}{m - i}, \frac{\lambda}{2}d\right\} \, di + \lambda dq
    \end{align}
    is monotonically increasing in \(q\) in the interval \([\frac{\lambda - 1}{\lambda}m, m]\).
\end{observation}

First, we will analyze the integral part \(I(q) = \int_{0}^{m - q}\max\left\{\lambda d - \frac{dm - \frac{\lambda}{2}d(m -q)}{m - i}, \frac{\lambda}{2}d\right\} \, di\) of both \cref{eq:W1:q:1,eq:W1:q:2}. For simplicity, let's define \(f(i, q) := \max\left\{\lambda d - \frac{dm - \frac{\lambda}{2}d(m -q)}{m - i}, \frac{\lambda}{2}d\right\}\). Then, using the Leibniz rule, we get:
\begin{align}
    \frac{dI(q)}{dq} & = f(m - q, q) \cdot \frac{d}{dq}(m - q) - \int_0^{m - q} \frac{\partial f(i, q)}{\partial q} \, di \nonumber \\
                     & = -f(m - q, q) - \int_0^{m - q} \frac{\partial f(i, q)}{\partial q} \, di \label{eq:derivative:I}
\end{align}

For the integral term, we need to perform a piecewise analysis due to the maximum function. Let \(g(i, q) := \lambda d - \frac{dm - \frac{\lambda}{2}d(m -q)}{m - i}\) and \(h = \frac{\lambda}{2}d\).
\begin{align*}
    \frac{\partial g(i, q)}{\partial q} = \frac{\lambda}{2} \cdot \frac{d}{m - i}
\end{align*}
Since \(f(i, q) = h\) is constant with respect to \(q\) if \(g(i, q) \leq h\), we get:
\begin{align}
    \frac{\partial f(i, q)}{\partial q} =
    \begin{cases}
        0,                                       & \text{if } g(i, q) \leq h \\
        \frac{\lambda}{2} \cdot \frac{d}{m - i}, & \text{otherwise}
    \end{cases}
    \label{eq:derivative:f}
\end{align}

Thus, the value of \(f(i, q)\) and its derivative are both positive. With this and \cref{eq:derivative:I}, we get:
\begin{align*}
    \frac{dI(q)}{dq} = -f(m - q, q) - \int_0^{m - q} \frac{\partial f(i, q)}{\partial q} \, di < 0
\end{align*}
Therefore \(I(q)\) is monotonically decreasing in \(q\). Consequently, we can directly see that \cref{eq:W1:q:1} is monotonically decreasing in \(q\) since the additional term \((\lambda - 1)dm\) is constant with respect to \(q\).

To prove \cref{obs:W1:q:2}, we have to show that the derivative of \cref{eq:W1:q:2} is positive.
\begin{align}
    \label{eq:W1:q:2:derivative}
    \frac{dI(q)}{dq} + \lambda d > 0
\end{align}

We continue to further analyze the first part.
\begin{align}
    \frac{dI(q)}{dq} & = -f(m - q, q) - \int_0^{m - q} \frac{\partial f(i, q)}{\partial q} \, di                                                                                     \nonumber                                                                                                       \\
                     & \geq -f(m - q, q) - \int_0^{m - q} \frac{\lambda}{2} \cdot \frac{d}{m - i} \, di                                                                                        &  & \text{(\cref{eq:derivative:f})}                                       \nonumber                  \\
                     & = -f(m - q, q) - \frac{\lambda}{2}d \cdot \ln\left(\frac{m}{q}\right)                                      \nonumber                                                                                                                                                          \\
                     & \geq -f(m - q, q) - \frac{\lambda}{2}d \cdot \ln\left(\frac{\lambda}{\lambda - 1}\right)                                                                                &  & \text{(\(q \geq \frac{\lambda - 1}{\lambda}m\); \cref{obs:W1:q:2})} \label{eq:W1:q:I:derivative}
\end{align}

Let's analyze \(f(i, q)\) more closely.
\begin{subsubcase}[\(\frac{\lambda}{2}d \geq \lambda d - \frac{dm - \frac{\lambda}{2}d(m - q)}{m - i}\)]
    In this case \(f(m - q, q) = \frac{\lambda}{2}d\). Thus, with \cref{eq:W1:q:I:derivative}, we get:
    \begin{align*}
        \frac{dI(q)}{dq} + \lambda d & > -\frac{\lambda}{2}d - \frac{\lambda}{2}d \cdot \ln\left(\frac{\lambda}{\lambda -1}\right) + \lambda d &  & \text{(\cref{eq:W1:q:I:derivative})}  \\
                                     & = \frac{\lambda}{2}d -             \frac{\lambda}{2}d \cdot \ln\left(\frac{\lambda}{\lambda -1}\right)                                             \\
                                     & > 0                                                                                                     &  & \text{(\(\lambda \leq \frac{3}{2}\))}
    \end{align*}
    Thus, we have shown \cref{eq:W1:q:2:derivative} in this case.
\end{subsubcase}
\begin{subsubcase}[\(\frac{\lambda}{2}d < \lambda d - \frac{dm - \frac{\lambda}{2}d(m -q)}{m - i}\)]
    In this case:
    \begin{align*}
        \frac{dI(q)}{dq} + \lambda d & > -f(m - q, q) - \frac{\lambda}{2}d \cdot \ln\left(\frac{\lambda}{\lambda - 1}\right) + \lambda d                                    &  & \text{(\cref{eq:W1:q:I:derivative})} \\
                                     & =
        -\left(\lambda d - \frac{dm - \frac{\lambda}{2}d(m -q)}{q}\right) - \frac{\lambda}{2}d \cdot \ln\left(\frac{\lambda}{\lambda - 1}\right) + \lambda d                                                          \\
                                     & = \frac{dm \left(2 - \lambda\right)}{2q} + \frac{\lambda}{2}d - \frac{\lambda}{2}d \cdot \ln\left(\frac{\lambda}{\lambda - 1}\right)                                           \\
                                     & > \frac{dm\left(2 - \lambda\right)}{2q}                                                                                                                                        \\
                                     & > 0                                                                                                                                  &  & \text{(\(\lambda < 2\))}
    \end{align*}
    Thus, \cref{eq:W1:q:2:derivative} holds in this case as well.
\end{subsubcase}

Since, in both cases, the derivative is positive in the interval \([\frac{\lambda - 1}{\lambda}m, m]\), \cref{obs:W1:q:2} holds.

With \cref{obs:W1:q:1} and \cref{obs:W1:q:2}, we know that the function \(W(q) = \max(a(q), b(q))\) has a minimum at the switching point of the second max-term, which is at \(q = \frac{\lambda - 1}{\lambda}m\).

\subsection{Case 1 (\(m < 6\))}
\label{sec:app:repair2-2:correctness:case2}
For this case of the proof, we assume \(m < 6\). For these instances, we explicitly calculate a lower bound on the work of the overall schedule to show the contradiction to \(\mathcal{W}_1 + \mathcal{W}_2 \leq dm\).

To simplify the calculations of the numerical values in this case, we will use \(\lambda = 1.45\), which is a bit stronger than what we need to show.

We prepare the proof with the following observation:
\begin{observation}
    \label{lem:upper:q:k1}
    If \(\lambda \geq 1.45\), then the total number of idle machines in $S_1$ is less than $q < \frac{11}{29}m$.
\end{observation}
\begin{proof}
    Note that we have $\mathcal{W}_2 > \frac{145}{100}dq$, by \cref{lem:repair-properties}.(\ref{lem:repair-properties:2}). With $\mathcal{W}_1 + \mathcal{W}_2 \leq dm$ and \cref{lem:repair-properties}.(\ref{lem:repair-properties:3}), we get:
    \begin{align*}
        \mathcal{W}_1 + \mathcal{W}_2                                 & \leq dm                            \\
        \Leftrightarrow \frac{145}{200}d(m_1 - q) + \frac{145}{100}dq & < dm                               \\
        \Leftrightarrow \frac{145}{200}q                              & < \frac{55}{200}m                  \\
        \Leftrightarrow q                                             & < \frac{55}{145}m = \frac{11}{29}m
    \end{align*}
\end{proof}

If \(j_0\) cannot be placed on \(m - i\) machines, then the load \(L_i\) on machine \(i \in \{0, \dots, m - 1\}\) is at least
\begin{align*}
    L_i > \begin{cases*}
              g(i)                                & if \(g(i) \geq \frac{\lambda}{2}d\) \\
              \max\{\lambda d - L_{i - 1}, g(i)\} & otherwise
          \end{cases*}
\end{align*}
\cref{tab:num_m1<6} shows the minimal work for all \(m \in \{3, 4, 5\}\), where \(\mathcal{W}_1 = \sum_{i}{L_i}\), and \(\mathcal{W}_2 = \max\{0.45dm, 1.45dq\}\). Notice that $q = 1$ and $m > 2$, since $m/6 < q < \frac{11}{29}m$ (\cref{lem:upper:q:k1}).
\begin{table}[h]
    \centering
    \begin{tabular}{|c|c|c|}
        \hline
        $m, q$ & Minimal Machine Loads (\(L_i\))                                     & $(\mathcal{W}_1 + \mathcal{W}_2) / m$ \\
        \hline
        3, 1   & ['0.93333', '0.67500', '0.00000']                                   & \underline{1.01944}                   \\
        4, 1   & ['0.99375', '0.84167', '0.60833', '0.00000']                        & \underline{1.06094}                   \\
        5, 1   & ['\underline{1.03000}', '0.92500', '0.75000', '0.70000', '0.00000'] & \underline{1.13100}                   \\
        \hline
    \end{tabular}
    \caption{Numerical Analysis for $m < 6$. All values are multiples of \(d\) and contradictional values are underlined.}
    \label{tab:num_m1<6}
\end{table}

\section{Adding the Small Jobs}
\label{sec:app:small-jobs}
To complete the dual approximation, we need to add the previously ignored small jobs \(J_S\). To this end, we recall that the total work on the obtained shelf schedule is at most \(md - \mathcal{W}_S\).

\begin{algorithm}[H]
    \caption{Adding the small jobs}\label{alg:addSmall}
    \begin{algorithmic}[1]
        \Require Shelf-Schedule for \(J_B\)
        \For{\textbf{each} $j \in J_S$}
        \State Schedule job $j$ on the least loaded machine.
        \EndFor
    \end{algorithmic}
\end{algorithm}

\begin{lemma}
    \label{lem:addsmall}
    Given a feasible contiguous schedule of length at most \(\lambda d\) for \(J_B\) with a work area of at most \(md - \mathcal{W}_S\), \cref{alg:addSmall} delivers a feasible contiguous schedule with makespan at most \(\lambda d\) for the entire instance \(J_B \cup J_S\), for any \(\lambda \geq \frac{10}{7}\).
\end{lemma}
\begin{proof}
    Consider a schedule, as constructed in \cref{sec:algorithm}. We schedule all jobs in shelf 1 with a starting time of \(0\) and all jobs in shelf 2 such that they complete at exactly \(\lambda d\). This way the idle times on each machine are uninterrupted.

    The load of each machine is defined by the sum of execution times of the jobs scheduled on it. By definition, the idle time of any machine \(i\) is equal to \(\lambda d\) minus the load on that machine.

    Now suppose \cref{alg:addSmall} does not deliver a feasible \(\lambda d\) schedule. In this case, there must exist a machine \(i\) with a load strictly greater than \(\lambda d\) after adding a small job. This implies that the load on that machine, before adding the small job, was greater than \(d\), since all small jobs have an execution time of at most \(\frac{3}{7}d\) on one machine, by definition. With \(i\) being the least loaded machine, this contradicts the assumption that the total work area of this schedule is at most \(md\).
\end{proof}

\end{document}